\documentclass[review,hidelinks,onefignum,onetabnum]{siamart220329}
\usepackage{fullpage,graphicx,subfigure,mathpazo,color}
\usepackage{amsmath,amscd,tikz,mathrsfs}
\usepackage[normalem]{ulem}
\usepackage{amssymb}
\usepackage{epstopdf}
\usepackage{tikz}
\usepackage{setspace}
\usepackage{cases}

\newcommand{\ii}{\mathrm{i}}
\newcommand{\ee}{\mathrm{e}}
\newcommand{\dd}{\mathrm{d}}

\newtheorem{rhp}{Riemann-Hilbert Problem}

\title{Inverse scattering transform for the integrable fractional derivative nonlinear Schr\"odinger equation\thanks{ {\bf Funding:}
		Liming Ling is supported by the National Natural Science Foundation of China (Grant No. 12122105); Xiaoen Zhang is supported by the National Natural Science Foundation of China (Grant No.12101246).}}

\author{Ling An\thanks{School of Mathematics, South China University of Technology, Guangzhou, China 510641
		(\email{maal@mail.scut.edu.cn}).}
	\and Liming Ling\thanks{School of Mathematics, South China University of Technology, Guangzhou, China 510641
		(\email{linglm@scut.edu.cn}).}
	\and Xiaoen Zhang\thanks{School of Mathematics, South China University of Technology, Guangzhou, China 510641
		(\email{zhangxiaoen@scut.edu.cn}).}}

\begin{document}

\maketitle
\begin{abstract}
In this paper, we explore the integrable fractional derivative nonlinear Schr\"odinger (fDNLS) equation by using the inverse scattering transform. Firstly, we start from the recursion operator and obtain a formal fDNLS equation. Then the inverse scattering problem is formulated and solved through the matrix Riemann-Hilbert problem. Subsequently, we give the explicit form of the fDNLS equation according to the properties of squared eigenfunctions, such as squared eigenfunctions are the eigenfunctions of the recursion operator of the integrable equations. The reflectionless potential with a simple pole for the zero boundary condition is carried out explicitly by means of determinants. Finally, for the fractional one-soliton solution, we analyze the wave propagation direction and the effect of the small fractional parameter $\epsilon$ on the wave. The fractional one-soliton solution has been verified rigorously. In addition, we also analyze the fractional rational solution obtained by taking the limit of the fractional one-soliton solution.
\end{abstract}



\section{Introduction}
Fractional calculus has a very long history \cite{oldham1974fractional,miller1993introduction,podlubny1999fractional}, originating from some conjectures of Leibniz and Euler. Fractional differential equations (FDEs) have been widely used to describe various physical effects, such as abnormal dispersion \cite{gazizov2009symmetry,yang2017new}, long-time behavior, subthreshold neural propagation \cite{magin2004fractional}, and so on \cite{kilbas2006theory,tarasov2011fractional}. Moreover, FDEs have been divided into many types according to the different definitions of fractional derivatives. Taking the well-known nonlinear Schr\"{o}dinger (NLS) equation as an example, there are several fractional forms \cite{al2018high,li2020vortex,li2021symmetry}. While it should be noted that these fractional equations are not integrable in the sense of inverse scattering transform (IST), which makes the obtained FDEs not have as good properties as the integrable equations.

In $2022$, Ablowitz, Been, and Carr proposed a new type of fractional equation, the fractional NLS equation and the fractional Korteweg-deVries (KdV) equation which are integrable in the sense of IST  \cite{ablowitz2022fractional}. The authors defined the fractional operator based on the Riesz fractional derivative \cite{Agrawal_2007}, which also can be called Riesz transform \cite{Riesz1949} or fractional Laplacian \cite{lischke2020fractional}, and a spectral representation for the fractional operator is then obtained by using the completeness of squared eigenfunctions. Then they claimed that this type of fractional equation could be applied to the whole Ablowitz-Kaup-Newell-Segur (AKNS) system \cite{ablowitz2022integrable}. Subsequently, the fractional forms of the higher-order modified KdV equation \cite{zhang2022interaction}, the higher-order NLS equation \cite{weng2022dynamics}, and so on were studied. In \cite{yan2022multi}, the author proposed a new integrable multi-L{\'e}vy-index and mixed fractional nonlinear equations. In \cite{zhong2022data}, the authors studied the integrable fractional equations via deep learning with Fourier neural operator. In \cite{an2022nondegenerate}, the authors took the fractional coupled Hirota equation as an example to explore the fractional integrable equation with $3\times 3$ Lax equation. In \cite{ablowitz2022discrete}, the authors extended the fractional integrable partial differential equations to the discrete NLS equation. However, the above integrable fractional equations are all related to the (discrete) AKNS system. Another significant integrable model named the derivative nonlinear Schr\"{o}dinger (DNLS) equation: 
	\begin{equation}\label{4-classicalDNLS}
		\ii q_{t}+q_{xx}+\ii(|q|^{2}q)_{x}=0,
	\end{equation}
which belongs to the Kaup-Newell (KN) system, plays a crucial role in the field of integrable systems. So the corresponding fractional extension of the KN system will be an interesting and natural problem in the theory of the fractional integrable system. In the equation \eqref{4-classicalDNLS}, $|q|=(qq^{*})^{\frac{1}{2}}$, $q^{*}$ is the complex conjugate of $q$, and the subscripts stand for partial derivative.

The DNLS equation was first proposed in $1971$ by Rogister \cite{rogister1971parallel}, and later derived by Mjolhus \cite{mjolhus1976modulational}, Mio, and others \cite{mio1976modified} in $1976$. This equation has many vital applications in different fields, such as long-wavelength dynamics of dispersive Alfv\'{e}n waves in the plasma physics field \cite{rogister1971parallel}, the subpicosecond or femtosecond pulse in a single-mode fiber \cite{ohkuma1987soliton}, the propagation of nonlinear pulses in optical fibers \cite{agrawal2001applications}, and so on. In mathematics, the DNLS equation also received much attention. Kaup and Newell gave the Lax pair of this equation and solved it using the IST \cite{kaup1978exact}. Furthermore, there are many other methods, such as the Hirota method \cite{kakei1995bilinearization}, the Darboux transformation \cite{xu2011darboux,guo2013high,khalique2022first}, and so on, are also used to solve the equation \eqref{4-classicalDNLS}. Based on these results already obtained for the DNLS equation, we want to explore the integrable fractional extension of the DNLS equation.

The organization of this paper is as follows. In Sec.$2$, we associate the KN spectral problem with a family of integrable nonlinear equations by introducing the recursion operator $\mathcal{L}$. Based on the idea in \cite{ablowitz2022fractional}, we extend the above set of integrable nonlinear equations to contain fractional integrable nonlinear equations and give the operator function $\mathcal{F}_{fd}(\mathcal{L})$ corresponding to the fractional DNLS (fDNLS) equation by using the dispersion relation. Note that for the fractional integrable equation in the sense of IST, the spectral matrix ${\bf V}(\lambda;x,t)$ related to time $t$ can not be written in a closed form, while some constraints should be added to it. Then, we use the IST to analyze some properties of eigenfunctions and scattering matrix, which help us find a completeness relation for the squared eigenfunctions of the fDNLS equation. The completeness relation of the squared eigenfunctions provides a spectral representation of the recursion operator $\mathcal{L}$, which corresponds to the fDNLS equation, thus allows us to give the explicit form of the fDNLS equation. In Sec.$3$, we explore the fractional $N$-soliton solution of the fDNLS equation. For the fractional one-soliton solution, we analyze the wave peak, the moving direction of the wave, and the influence of the small parameter $\epsilon$ in fractional power on wave propagation. We also analyze the fractional rational solution by taking the limit of the fractional one-soliton solution. More importantly, we provide detailed and rigorous proof of the fractional one-soliton solution.

\section{The fDNLS equation and its IST}\label{3-sec.2}
In this section, we will give the explicit form of the fDNLS equation, and solve this equation by IST. The construction of an integrable fractional equation requires three key elements: the general evolution equation which can be solved by IST, the anomalous dispersion relation, and the completeness of squared eigenfunctions. 

The anomalous dispersion relation is related to the recursion operator, so we need to give the recursion operator of the DNLS equation first. From a matrix spectral problem with arbitrary parameters, we can construct the generalized DNLS hierarchy \cite{fan-2001}. Then the recursion operator of the DNLS equation can be found. For convenience, we introduce three Pauli's spin matrices:
\begin{equation*}
	\sigma_{1}=\begin{bmatrix}
		0&1\\
		1&0
	\end{bmatrix},\ \ \ \ 
	\sigma_{2}=\begin{bmatrix}
		0&-\ii\\
		\ii&0
	\end{bmatrix},\ \ \ \ 
	\sigma_{3}=\begin{bmatrix}
		1&0\\
		0&-1
	\end{bmatrix}.
\end{equation*}
Now we consider the following spectral problem:
\begin{equation}\label{4-Lax-Phi}
	\begin{split}
		&\Phi_{x}={\bf U}\Phi,\ \ \ \ {\bf U}(\lambda;x,t)=-\ii(\lambda^{2}-\alpha qr)\sigma_{3}+\lambda {\bf Q}(x,t),\ \ \ {\bf Q}(x,t)=\begin{bmatrix}
			0 & q\\[2pt]
			r & 0
		\end{bmatrix},\\
		&\Phi_{t}={\bf V}\Phi,\ \ \ \ {\bf V}(\lambda;x,t)=\begin{bmatrix}
			V_{1}&V_{2}\\[4pt]
			V_{3}&-V_{1}
		\end{bmatrix},
		\end{split}
	\end{equation}
where $\Phi(\lambda;x,t)$ is the wave function, $\lambda\in\mathbb{C}$ is the spectral parameter, $\alpha\in\mathbb{R}$, $q=q(x,t)$ and $r=r(x,t)$ are potential functions, $V_{j}(\lambda;x,t),\ j=1,2,3$ are the quantities depending on $q,\ r,$ and their derivatives.

The compatibility condition or the zero curvature equation of \eqref{4-Lax-Phi}, reads
\begin{equation}\label{4-ZCE}
	{\bf U}_{t}-{\bf V}_{x}+[{\bf U},{\bf V}]=0,\ \ \ \ [{\bf U},{\bf V}]\equiv {\bf UV-VU},
\end{equation}
which implies
\begin{equation}\label{4-V}
	\begin{split}
		\ii\alpha(qr)_{t}&=V_{1x}-\lambda(qV_{3}-r V_{2}),\\[2pt]
		\lambda q_{t}&=V_{2x}+2\ii(\lambda^2-\alpha qr)V_{2}+2\lambda qV_{1},\\[2pt]
		\lambda r_{t}&=V_{3x}-2\ii(\lambda^2-\alpha qr)V_{3}-2\lambda rV_{1}.
	\end{split}
	\end{equation}
Combining with \eqref{4-V}, there is
\begin{equation}\label{4-rewritten-4V2+4V3}
	\lambda\begin{bmatrix}
		q\\[2pt]r
	\end{bmatrix}_{t}=2\lambda V_{10}\sigma_{3}\begin{bmatrix}
		q\\[2pt]r
	\end{bmatrix}+(\mathcal{L}_{1}\mathcal{L}_{2}+2\ii \lambda^{2}\mathcal{L}_{3})\sigma_{3}\begin{bmatrix}
		V_{2}\\[2pt]V_{3}
	\end{bmatrix},
\end{equation}
where $V_{10}$ is an integration constant,
\begin{equation*}
	\begin{split}
		\mathcal{L}_{1}=\mathbb{I}+2\ii \alpha\sigma_{3}\begin{bmatrix}
			q\\[2pt]r
		\end{bmatrix}\partial^{-1}\begin{bmatrix}
			r,&q
		\end{bmatrix},\ \ \ \ \mathcal{L}_{2}=\sigma_{3}\partial_{x}-2\ii\alpha qr\mathbb{I},\ \ \ \ \mathcal{L}_{3}=\mathbb{I}+\ii(2\alpha+1)\sigma_{3}\begin{bmatrix}
			q\\[2pt]r
		\end{bmatrix}\partial^{-1}\begin{bmatrix}
			r,&q
		\end{bmatrix}.
	\end{split}
\end{equation*}
We assume $V_{10}=-2\ii(-\lambda^{2})^{n}$,
\begin{equation*}
	\begin{bmatrix}
		V_{2}\\[2pt]V_{3}
	\end{bmatrix}=\sum_{j=1}^{n}(-1)^{n-j}\begin{bmatrix}
		V_{2j}\\[2pt]V_{3j}
	\end{bmatrix}\lambda^{2(n-j)+1},
\end{equation*}
substitute them into \eqref{4-rewritten-4V2+4V3}, and group terms according to their power of $\lambda$. Then we can get the set of integrable nonlinear equations through the iterating calculation,
\begin{equation}\label{4-evolution eq}
	\begin{bmatrix}
		q\\[2pt]r
	\end{bmatrix}_{t}=-4\ii\mathcal{F}(\mathcal{L})\begin{bmatrix}
		q\\[2pt]-r
	\end{bmatrix},\ \ \ \ \mathcal{F}(\mathcal{L})=\mathcal{L}^{n},
\end{equation}
where
\begin{equation*}
	\begin{split}
	\mathcal{L}=\frac{1}{2\ii}\mathcal{L}_{1}\mathcal{L}_{2}\mathcal{L}_{3}^{-1}=&\frac{1}{2}\sigma_{2}\sigma_{1}\partial_{x}-\left(\alpha+\frac{1}{2}\right){\bf u}_{x}\partial^{-1}{\bf u}^{\top}\sigma_{1}-\left(2\alpha+\frac{1}{2}\right){\bf u}{\bf u}^{\top}\sigma_{1}+\frac{1}{2}\alpha {\bf u}^{\top}\sigma_{1}{\bf u}\mathbb{I}\\
	&-\alpha\sigma_{3}{\bf u}\partial^{-1}{\bf u}_{x}^{\top}\sigma_{1}\sigma_{3}+\alpha\left(2\alpha+\frac{1}{2}\right)\sigma_{2}\sigma_{1}{\bf u}\partial^{-1}{\bf u}^{\top}\sigma_{1}{\bf u}{\bf u}^{\top}\sigma_{1},
\end{split}
\end{equation*}
${\bf u}(x,t)=\begin{bmatrix}
	q,&r
\end{bmatrix}^{\top}$, the superscript $^{\top}$ denotes the transpose. 
The system \eqref{4-evolution eq} is a generalized system, which can yield many integrable equations by choosing different parameters $n$ and $\alpha$. For example, the generalized DNLS hierarchy can be obtained by choosing $n=2$ and $r=\sigma q^{*}\ (\sigma=\pm 1)$. Moreover, the cases of $\alpha=0,\ -\frac{1}{2},\ -\frac{1}{4}$ correspond to the DNLS equation, the  Chen-Lee-Liu equation, and the Gerdjikov-Ivanov equation, respectively. Note that $\mathcal{F}(\mathcal{L})$ can be generalized to the more general form $\mathcal{F}(\lambda)$ by using the properties of squared eigenfunctions, and the function $\mathcal{F}(\lambda)$ can be associated with the linear dispersion relation of \eqref{4-evolution eq}. Without loss of generality, we will take the equation related to $q(x,t)$ in \eqref{4-evolution eq} as an example to deduce the relation between $\mathcal{F}(\lambda)$ and the linear dispersion relation. The linearized form of the equation \eqref{4-evolution eq} can be given by
	\begin{equation*}
		\begin{bmatrix}
			q\\[2pt]r
		\end{bmatrix}_{t}=-4\ii\mathcal{F}\left(\frac{1}{2}\sigma_{2}\sigma_{1}\partial_{x}\right)\begin{bmatrix}
			q\\[2pt]-r
		\end{bmatrix},\ \ \ \ \mathcal{F}\left(\frac{1}{2}\sigma_{2}\sigma_{1}\partial_{x}\right)=\left(\frac{1}{2}\sigma_{2}\sigma_{1}\partial_{x}\right)^{n},
	\end{equation*}
	then 
	\begin{equation}\label{4-evolution-q-linear}
		q_{t}(x,t)=-4\ii\left(-\frac{\ii}{2}\partial_{x}\right)^{n}q(x,t).
	\end{equation}
	Then we substitute the formal solution  $q(x,t)\thicksim\ee^{\ii(\lambda x-\omega(\lambda)t)}$ into \eqref{4-evolution-q-linear}, which yields
	\begin{equation}\label{4-relation-F-dis}
		\mathcal{F}\left(\frac{\lambda}{2}\right)=\frac{1}{4}\omega\left(\lambda\right).
	\end{equation}
	
	In this paper, we will study the fractional equation in the KN system based on the DNLS equation \eqref{4-classicalDNLS}, in which explicit forms for ${\bf U}(\lambda;x,t)$ and ${\bf V}(\lambda;x,t)$ are given by
	\begin{equation*}
		{\bf U}=-\ii\lambda^{2}\sigma_{3}+\lambda {\bf Q},\ \ \ \ {\bf V}=-2\ii\lambda^{4}\sigma_{3}+2\lambda^{3}{\bf Q}-\ii\lambda^{2}\sigma_{3}{\bf Q}^{2}-\lambda(\ii\sigma_{3}{\bf Q}_{x}-{\bf Q}^{3}),
	\end{equation*}
	where $r(x,t)=-q^{*}(x,t)$. Let us rewrite the recursion operator $\mathcal{L}$ which corresponds to the DNLS equation 
	\begin{equation*}
		\mathcal{L}=\frac{1}{2}\sigma_{2}\sigma_{1}\partial_{x}-\frac{1}{2}{\bf u}_{x}\partial^{-1}_{-}{\bf u}^{\top}\sigma_{1}-\frac{1}{2}{\bf u}{\bf u}^{\top}\sigma_{1},\ \ \ \ \partial_{-}^{-1}=\int_{-\infty}^{x}\dd y.
	\end{equation*}
	The DNLS equation corresponds to the operator function $\mathcal{F}_{d}(\mathcal{L})=\mathcal{L}^{2}$. Based on the equation \eqref{4-relation-F-dis}, there is $\omega_{d}=\lambda^{2}$. In \cite{ablowitz2022fractional}, we can find that the linear dispersion relation of the NLS equation is $\omega_{N}(\lambda)=-\lambda^{2}$, then the fractional NLS equation can be obtained according to $\omega_{fN}(\lambda)=-\lambda^{2}|\lambda|^{\epsilon},\ \epsilon\in[0,1)$. Following this rule, we can assume that the anomalous dispersion relation of fDNLS equation is $\omega_{fd}(\lambda)=\lambda^{2}|\lambda|^{\epsilon},\ \epsilon\in[0,1)$. Then the linearized form of the fDNLS equation is
	\begin{equation*}
		\ii q_{t}+|-\partial_{x}^{2}|^{\frac{\epsilon}{2}}q_{xx}=0,
	\end{equation*}
	where $|-\partial^{2}|^{\epsilon},\ \epsilon\in[0,1)$ is called the Riesz fractional derivative \cite{ablowitz2022fractional}. Combining with the relation \eqref{4-relation-F-dis},
	\begin{equation*}
		\mathcal{F}_{fd}(\lambda)=\frac{1}{4}\omega_{fd}(2\lambda)=\lambda^{2}|2\lambda|^{\epsilon},\ \  \epsilon\in[0,1).
	\end{equation*}
	This will lead to the operator function $\mathcal{F}_{fd}(\mathcal{L})$ which corresponds to the fDNLS equation,
	\begin{equation*}
		\mathcal{F}_{fd}(\mathcal{L})=\mathcal{L}^{2}|2\mathcal{L}|^{\epsilon},\ \ \epsilon\in[0,1).
	\end{equation*}
	 We will give the explicit form of the fDNLS equation in
section \ref{4-subsection-fDNLS}.

\subsection{Direct scattering}\label{sec-direct scattering}
In the limit $|x|\to\infty$, we assume that the potential function $q(x,t)$ is sufficiently smooth and rapidly tends to zero. It is significant that the matrix ${\bf V}(\lambda;x,t)$ cannot be written precisely for the fractional integrable equation. But we need to impose a constraint ${\bf V}{\to}-2\ii\mathcal{F}_{fd}(\lambda^{2})\sigma_{3},$ as $|x|{\to}\infty$ on it.

The Lax pair of the fDNLS equation is 
\begin{equation}\label{4-Lax pair-DNLS}
	\Phi_{x}={\bf U}\Phi=(-\ii\lambda^{2}\sigma_{3}+\lambda {\bf Q})\Phi,\ \ \ \  \Phi_{t}={\bf V}\Phi,\ \ \ \ {\bf Q}(x,t)=\begin{bmatrix}
		0&q(x,t)\\[2pt]
		-q^{*}(x,t)&0
	\end{bmatrix}.
\end{equation}
Then we can derive the asymptotic behavior:
\begin{equation}\label{4-Phi-asy-be}
	\Phi^{\pm}\sim\ee^{-\ii\lambda^{2}\sigma_{3}x},\ \ \ x\to\pm\infty,
\end{equation}
where $\Phi^{\pm}(\lambda;x,t)=
\begin{bmatrix}
	\phi^{\pm}_{1},&\phi^{\pm}_{2}
\end{bmatrix}$, the superscripts $^{\pm}$ refer to the cases of $x\to\pm\infty$, respectively. Based on Abel's formula, we know that the determinants of $\Phi^{\pm}(\lambda;x,t)$ are independent of $x$, since ${\rm tr}{\bf Q}(x,t)=0$. Then we have
\begin{equation}\label{4-det-phi}
	\det\Phi^{\pm}=\lim_{x\to\pm\infty}\det\Phi^{\pm}=1.
\end{equation}
Next we consider the Jost solutions $\Psi^{\pm}(\lambda;x,t)=\Phi^{\pm}(\lambda;x,t) \ee^{\ii\lambda^{2}\sigma_{3}x}.$ $\Psi^{\pm}(\lambda;x,t)$ have the asymptotic properties: $\Psi^{\pm}(\lambda;x,t)\to\mathbb{I}$, as $x\to\pm\infty$. Then we can get the equations of $\Psi^{\pm}(\lambda;x,t)$, which is equivalent to the first equation in \eqref{4-Lax pair-DNLS}
\begin{equation}\label{4-Psi}
	\Psi_{x}^{\pm}=-\ii\lambda^{2}[\sigma_{3},\Psi^{\pm}]+\lambda{\bf Q}\Psi^{\pm}.
\end{equation}
Moreover, the Volterra integral equations for $\Psi^{\pm}(\lambda;x,t)$ are given by
\begin{equation}\label{4-Psi-Volterra}
	\Psi^{\pm}(\lambda;x,t)=\mathbb{I}+\lambda\int_{\pm\infty}^{x}\ee^{-\ii\lambda^{2}(x-y){\rm ad}\widehat{\sigma}_{3}}{\bf Q}(y,t)\Psi^{\pm}(\lambda;y,t)dy,
\end{equation}
where $\Psi^{\pm}(\lambda;x,t)=\begin{bmatrix}
	\psi^{\pm}_{1},&\psi^{\pm}_{2}
\end{bmatrix}$, $\ee^{{\rm ad}\widehat{\sigma}_{3}}{\bf X}=\ee^{\sigma_{3}}{\bf X}\ee^{-\sigma_{3}}$ with ${\bf X}$ being a $2\times 2$ matrix. The large-$\lambda$ expansions of the Jost solutions $\Psi^{\pm}(\lambda;x,t)$ are given by
\begin{equation*}
	\Psi^{\pm}(\lambda;x,t)=\exp\left(-\frac{\ii\sigma_{3}}{2}\int_{\pm\infty}^{x}|q(y,t)|^{2}dy\right)+\mathcal{O}(\lambda^{-1}).
\end{equation*}
According to the Volterra integral equations \eqref{4-Psi-Volterra}, the relevant properties of $\Psi^{\pm}(\lambda;x,t)$ or $\Phi^{\pm}(\lambda;x,t)$ can be analyzed, which are summarized in the following proposition. Similar results have been reported in \cite{zhang2020derivative}, so we will not prove them.
\begin{proposition}{\rm (see Lemma $1$ in \cite{pelinovsky-2018})}\label{4-prop-phi-analytic}
	Assume that  $q(x,t)\in\mathrm{L}^{1}(\mathbb{R})\cap\mathrm{L}^{3}(\mathbb{R})$ and $\partial_{x}q(x,t)\in\mathrm{L}^{1}(\mathbb{R})$, then there exist unique
	solutions satisfying the Volterra integral equations \eqref{4-Psi-Volterra} for every $\lambda\in\Sigma$, and the Jost solutions $\Psi^{\pm}(\lambda;x,t)$ possesses the following properties:
	\begin{itemize}
		\item  The column vectors $\psi^{-}_{1}$ and $\psi^{+}_{2}$ are analytic for $\lambda\in D^{+}$ and continuous for $\lambda\in D^{+}\cup\Sigma$,
		\item  The column vectors $\psi^{+}_{1}$ and $\psi^{-}_{2}$ are analytic for $\lambda\in D^{-}$ and continuous for $\lambda\in D^{-}\cup\Sigma$,
	\end{itemize}
	where $\Sigma=\mathbb{R}\cup\ii\mathbb{R}$, $D^{\pm}$ are shown in Fig.\ref{fig:contour}.
\end{proposition}
\begin{figure}
	\centering
	\includegraphics[width=0.3\linewidth]{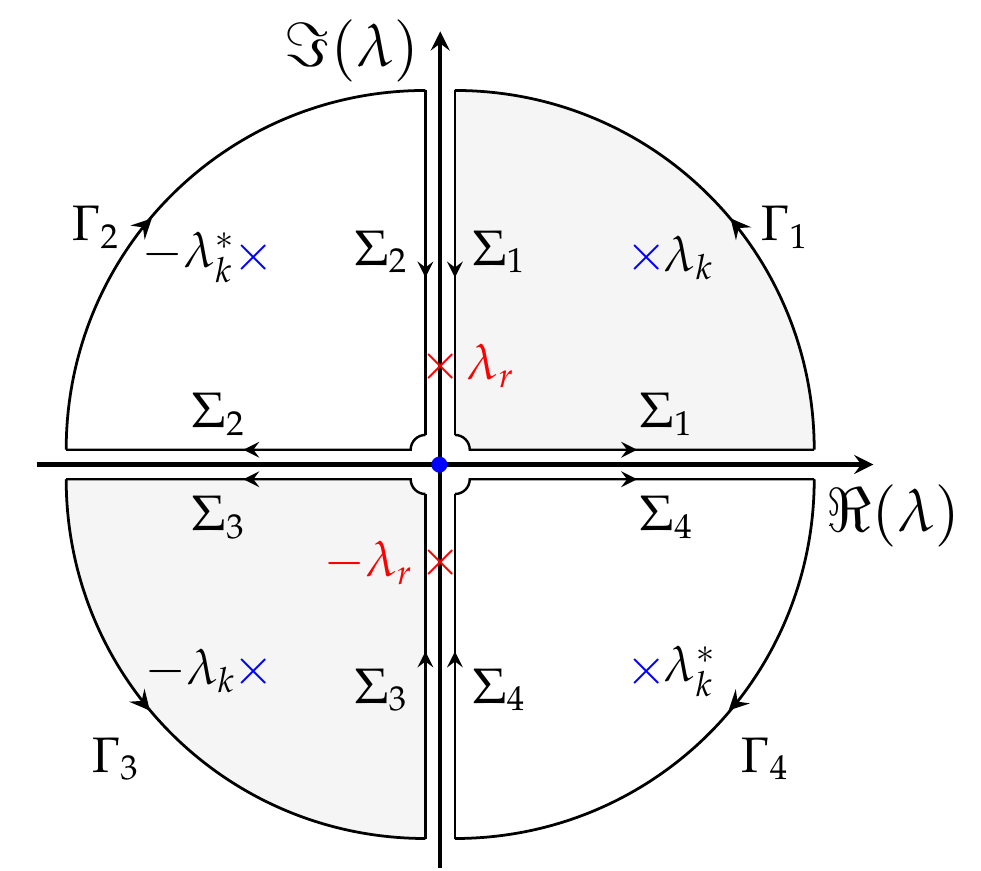}
	\caption{Complex $\lambda$-plane. The grey and white regions represent $D^{+}$ and $D^{-}$, respectively. $\Sigma_{+}=\Sigma_{1}+\Sigma_{3},\ \Sigma_{-}=\Sigma_{2}+\Sigma_{4}$.  $\Gamma_{+}=\Gamma_{1}+\Gamma_{3},\ \Gamma_{-}=\Gamma_{2}+\Gamma_{4}$. Discrete spectrum, and the integral path for IST of the fDNLS equation. $\lambda_{r}$ and $-\lambda_{r}$ are the discrete spectrum corresponding to the fractional rational solution.}
	\label{fig:contour}
\end{figure} 
Both $\Phi^{\pm}(\lambda;x,t)$ are the fundamental solutions of \eqref{4-Lax pair-DNLS}, then there exists a matrix ${\bf S}(\lambda;t)=\left(s_{ij}(\lambda;t)\right)_{i,j=1,2}$ between them obeying the relation:
\begin{equation}\label{4-Phi-S}
	\Phi^{-}(\lambda;x,t)=\Phi^{+}(\lambda;x,t){\bf S}(\lambda;t),\ \ \ \ \lambda\in\Sigma,
\end{equation}
and ${\bf S}(\lambda;t)$ is called the scattering matrix, $s_{ij}(\lambda;t),\ i,j=1,2$ are called the scattering coefficients. Moreover, $\det{\bf S}(\lambda;t){=}1$ can be deduced from \eqref{4-det-phi}. As usual, we define the reflection coefficients $\rho_{1}(\lambda;t)$ and $\rho_{2}(\lambda;t)$:
\begin{equation}\label{4-rho}
	\rho_{1}(\lambda;t)=\frac{s_{21}(\lambda;t)}{s_{11}(\lambda;t)},\ \ \ \ \rho_{2}(\lambda;t)=\frac{s_{12}(\lambda;t)}{s_{22}(\lambda;t)},\ \ \ \ \lambda\in\Sigma.
\end{equation}
In addition, the large-$\lambda$ expansion of ${\bf S}(\lambda;t)$ is
\begin{equation*}
	{\bf S}(\lambda;t)=\exp\left(-\frac{\ii\sigma_{3}}{2}\int_{-\infty}^{+\infty}|q(y,t)|^{2}dy\right)+\mathcal{O}(\lambda^{-1}).
\end{equation*}
By using \eqref{4-Phi-S}, there are
\begin{equation}\label{4-wr}
	\begin{split}
		&s_{11}(\lambda;t)={\rm Wr}(\phi_{1}^{-},\phi_{2}^{+}),\ \ \ \ s_{12}(\lambda;t)={\rm Wr}(\phi_{2}^{-},\phi_{2}^{+}),\\
		&s_{21}(\lambda;t)={\rm Wr}(\phi_{1}^{+},\phi_{1}^{-}),\ \ \ \ s_{22}(\lambda;t)={\rm Wr}(\phi_{1}^{+},\phi_{2}^{-}).
	\end{split}
\end{equation}
Therefore, $s_{11}(\lambda;t)$ and $s_{22}(\lambda;t)$ are analytic in $D^{\pm}$, respectively. Generally, the off-diagonal scattering
coefficients cannot be extended off the contour $\Sigma$.
\begin{proposition}\label{prop-sym}
	The fundamental solution $\Phi(\lambda;x,t)$, the Jost solution $\Psi(\lambda;x,t)$, and the scattering matrix ${\bf S}(\lambda;t)$ all have two symmetry reductions:
	\begin{itemize}
		\item $\Phi(\lambda;x,t)=\sigma_{2}\Phi^{*}(\lambda^{*};x,t)\sigma_{2}$,\ \ \ $\Psi(\lambda;x,t)=\sigma_{2}\Psi^{*}(\lambda^{*};x,t)\sigma_{2}$,\ \ \ ${\bf S}(\lambda;t)=\sigma_{2}{\bf S}^{*}(\lambda^{*};t)\sigma_{2}$.
		\item $\Phi(\lambda;x,t)=\sigma_{3}\Phi(-\lambda;x,t)\sigma_{3}$,\ \ \ \ $\Psi(\lambda;x,t)=\sigma_{3}\Psi(-\lambda;x,t)\sigma_{3}$,\ \ \ ${\bf S}(\lambda;t)=\sigma_{3}{\bf S}(-\lambda;t)\sigma_{3}$.
	\end{itemize}
\end{proposition}
\begin{proof}
	It is easy to find the symmetry reductions of the matrix ${\bf U}(\lambda;x,t)$, and therefore the symmetry reductions of $\Phi(\lambda;x,t)$ can be obtained according to the equation \eqref{4-Lax pair-DNLS}. And the symmetry reductions of $\Psi(\lambda;x,t)$ and ${\bf S}(\lambda;t)$ can also be directly deduced via equations \eqref{4-Psi} and \eqref{4-Phi-S}.
\end{proof}
According to the symmetry properties of ${\bf S}(\lambda;t)$ in proposition \ref{prop-sym}, we can derive $s_{11}(\lambda;t){=}s_{11}(-\lambda;t)$ ${=}s_{22}^{*}(\lambda^{*};t){=}s_{22}^{*}(-\lambda^{*};t)$. Then the zeros of $s_{11}(\lambda;t)$ appear in pairs, and we can suppose that $s_{11}(\lambda;t)$ has simple zeros defined by $\lambda_{n},\ n{=}1,2,\cdots,N$ in the ${\rm\uppercase\expandafter{\romannumeral1}}$ quadrant, and $\lambda_{n+N}{=}-\lambda_{n}$ in the ${\rm\uppercase\expandafter{\romannumeral3}}$ quadrant. That is to say, $s_{11}(\lambda_{j};t){=}0$, $s_{11}^{'}(\lambda_{j};t){\neq}0$, $j{=}1,2,\cdots,2N$, the superscript $^{'}$ denotes the partial derivative with respect to $\lambda$. Then $s_{22}(\lambda_{j}^{*};t){=}0,\ s_{22}^{'}(\lambda_{j}^{*};t){\neq}0,\ j{=}1,2,\cdots,2N$. So the discrete
spectrum can be defined by the set:
\begin{equation*}
	\Lambda=\big\{\lambda_{k},\ -\lambda_{k},\ \lambda_{k}^{*},\ -\lambda_{k}^{*}\big\}_{k=1}^{N},
\end{equation*}
whose distributions are shown in Fig.\ref{fig:contour}. Based on the equation \eqref{4-wr}, it can be seen that when $s_{11}(\lambda_{j};t)=0,\ j=1,2,\cdots,2N$,  $\phi_{1}^{-}(\lambda_{j};x,t)$ and $\phi_{2}^{+}(\lambda_{j};x,t)$ must be
proportional
\begin{equation}\label{relation-a}
	\phi_{1}^{-}(\lambda_{j};x,t)=a_{j}(\lambda_{j};t)\phi_{2}^{+}(\lambda_{j};x,t),\ \ \ \ j=1,2,\cdots,2N,
\end{equation}
where $a_{j}(\lambda_{j};t):=a_{j}\neq0$. Similarly, when $s_{22}(\lambda_{j}^{*};t){=}0,\ j{=}1,2,\cdots,2N$, there is
\begin{equation}\label{relation-b}
	\phi_{2}^{-}(\lambda_{j}^{*};x,t)=b_{j}(\lambda_{j}^{*};t)\phi_{1}^{+}(\lambda_{j}^{*};x,t),\ \ \ \ j=1,2,\cdots,2N,
\end{equation}
where $b_{j}(\lambda_{j}^{*};t):=b_{j}\neq0$. Furthermore, the relations  $a_{j}=-b_{j}^{*},\ a_{j+N}=-a_{j}$ can be derived by combining \eqref{relation-a}, \eqref{relation-b}, together with the symmetry reductions of $\Phi(\lambda;x,t)$ in proposition \ref{prop-sym}.

\subsection{Time evolution}

The time evolution of the scattering data can be obtained by analyzing the asymptotic behavior of the associated time evolution operator ${\bf V}(\lambda;x,t)$, which cannot be represented generally. From the section \ref{sec-direct scattering}, we know ${\bf V}{\to}-2\ii\mathcal{F}_{fd}(\lambda^{2})\sigma_{3}$ as $|x|{\to}\infty$, then 
	\begin{equation*}
		\begin{split}
			&s_{11}(\lambda;t)=s_{11}(\lambda;0),\ \ \ \  s_{12}(\lambda;t)=\ee^{4\ii\mathcal{F}_{fd}(\lambda^{2})t}s_{12}(\lambda;0),\\
			&s_{22}(\lambda;t)=s_{22}(\lambda;0),\ \ \ \ s_{21}(\lambda;t)=\ee^{-4\ii\mathcal{F}_{fd}(\lambda^{2})t}s_{21}(\lambda;0).\\
		\end{split}
	\end{equation*}
	In addition, there are
	\begin{equation*}
		a_{j}(\lambda_{j};t)=\ee^{-4\ii\mathcal{F}_{fd}(\lambda_{j}^{2})t}a_{j}(\lambda_{j};0),\ \ \ \ b_{j}(\lambda_{j}^{*};t)=\ee^{4\ii\mathcal{F}_{fd}(\lambda_{j}^{*2})t}b_{j}(\lambda_{j}^{*};0).
\end{equation*}

\subsection{Inverse scattering}
Now we consider the inverse problem in terms of the relation \eqref{4-Phi-S}. By reviewing the analytic properties of $\Psi^{\pm}(\lambda;x,t)$, we can define a sectional analytic matrix ${\bf M}(\lambda;x,t)$:
\begin{equation}\label{4-M-def}
	\begin{split}
		&{\bf M}_{+}(\lambda;x,t)=\begin{bmatrix}
			\dfrac{\psi^{-}_{1}(\lambda;x,t)}{s_{11}(\lambda;t)},&\psi^{+}_{2}(\lambda;x,t)
		\end{bmatrix},\ \ \ \ \lambda\in D^{+},\\[4pt]
		&{\bf M}_{-}(\lambda;x,t)=\begin{bmatrix}
			\psi^{+}_{1}(\lambda;x,t),&\dfrac{\psi^{-}_{2}(\lambda;x,t)}{s_{22}(\lambda;t)}
		\end{bmatrix},\ \ \ \ \lambda\in D^{-}.
	\end{split}
\end{equation}
Then we can formulate the following Riemann-Hilbert problem.
\begin{rhp}\label{rhp-1}
We can find the matrix ${\bf M}(\lambda;x,t)$ with the following properties:
\begin{itemize}
	\item ${\bf Analyticity:}$\ ${\bf M}_{\pm}(\lambda;x,t)$ are sectionally meromorphic in $D^{\pm}\backslash\Lambda$, and have the simple poles in $\Lambda$, whose principal parts of the Laurent series at each simple pole $\lambda_{k}$ or $\lambda_{k}^{*}$, are determined as
	\begin{equation*}
		\begin{split}
			&\underset{\lambda=\lambda_{k}}{\rm{Res}}\ {\bf M}(\lambda;x,t)=\begin{bmatrix}
				\dfrac{a_{k}\psi_{2}^{+}(\lambda_{k};x,t)}{s_{11}^{'}(\lambda_{k};t)} \exp\left(2\ii\left(\lambda_{k}^{2}x+2\mathcal{F}_{fd}(\lambda_{k}^{2})t\right)\right),&\ \ \ \ \ 0
			\end{bmatrix},\\[4pt]
			&\underset{\lambda=\lambda_{k}^{*}}{\rm{Res}}\ {\bf M}(\lambda;x,t)=\begin{bmatrix}
				0,&\dfrac{b_{k}\psi_{1}^{+}(\lambda_{k}^{*};x,t)}{s_{22}^{'}(\lambda_{k}^{*};t)}\exp\left(-2\ii\left(\lambda_{k}^{*2}x+2\mathcal{F}_{fd}(\lambda_{k}^{*2})t\right)\right)
			\end{bmatrix},
		\end{split}
	\end{equation*}
	where the superscript $^{'}$ denotes the partial derivative with respect to $\lambda$.
	\item ${\bf Jump\ condition:}$\ ${\bf M}(\lambda;x,t)$ satisfies the jump condition:
	\begin{equation*}
		{\bf M}_{+}(\lambda;x,t)={\bf M}_{-}(\lambda;x,t)\ \left(\mathbb{I}-{\bf J}(\lambda;x,t)\right),\ \ \lambda\in\Sigma,
	\end{equation*}
	where
	\begin{equation*}
		{\bf J}(\lambda;x,t)=\exp\left(-\ii\big(\lambda^{2}x+2\mathcal{F}_{fd}(\lambda^{2})t\big){\rm ad}\widehat{\sigma}_{3}\right)\begin{bmatrix}
			\rho_{1}(\lambda;t)\rho_{2}(\lambda;t)&\rho_{2}(\lambda;t)\\[4pt]
			-\rho_{1}(\lambda;t)&0
		\end{bmatrix}.
	\end{equation*}
	\item ${\bf Asymptotic\ behavior:}$\
	\begin{equation*}
		{\bf M}(\lambda;x,t)=\exp\left(\frac{\ii\sigma_{3}}{2}\int_{x}^{+\infty}|q(y,t)|^{2}dy\right)+\mathcal{O}(\lambda^{-1}),\ \ \ \ {\rm as}\ \lambda\to\infty.
	\end{equation*}
\end{itemize}		
\end{rhp}
To solve the above Riemann-Hilbert problem, we need to regularize it by subtracting the  pole contributions and the asymptotic behavior. So we define a new matrix ${\bf M}^{[1]}(\lambda;x,t)$ as follows:
\begin{equation*}
{\bf M}^{[1]}_{\pm}(\lambda;x,t){=}{\bf M}_{\pm}(\lambda;x,t){-}\exp\left(\frac{\ii\sigma_{3}}{2}\int_{x}^{+\infty}|q(y,t)|^{2}dy\right){-}\sum_{k=1}^{2N}\left(\dfrac{\underset{\lambda=\lambda_{k}}{\rm{Res}}{\bf M}(\lambda;x,t)}{\lambda-\lambda_{k}}{+}\dfrac{\underset{\lambda=\lambda_{k}^{*}}{\rm{Res}}{\bf M}(\lambda;x,t)}{\lambda-\lambda_{k}^{*}}\right),\ \lambda\in D^{\pm}.
\end{equation*}
Note that ${\bf M}^{[1]}_{\pm}(\lambda;x,t)$ are sectionally meromorphic in $D^{\pm}$, and 
${\bf M}^{[1]}(\lambda;x,t)=\mathcal{O}(\lambda^{-1})\ {\rm as}\ \lambda\to\infty$.
Moreover,
\begin{equation}\label{4-regular-Jump}
	{\bf M}^{[1]}_{+}(\lambda;x,t)-{\bf M}^{[1]}_{-}(\lambda;x,t)=-{\bf M}_{-}(\lambda;x,t){\bf J}(\lambda;x,t),\ \ \lambda\in\Sigma.
\end{equation}
Here we introduce the Cauchy projectors $\mathcal{P}_{\pm}$ over $\Sigma$ \cite{biondini2014inverse} defined by
\begin{equation*}
	\mathcal{P}_{\pm}[f](\lambda)=\frac{1}{2\pi\ii}\int_{\Sigma}\dfrac{f(\zeta)}{\zeta-(\lambda\pm\ii 0)}\dd\zeta,
\end{equation*}
where the notation $\lambda\pm\ii0$ indicates that the limit is taken from the left/right of $\lambda$ along the direction. Based on the Plemelj's formulae, there are $\mathcal{P}_{\pm}f_{\pm}=\pm f_{\pm}$, $\mathcal{P}_{+}f_{-}=\mathcal{P}_{-}f_{+}=0$, when $f_{\pm}$ are analytic in $D^{\pm}$, and are equal to $\mathcal{O}(\lambda^{-1})$ as $\lambda\to\infty$. And we introduce the notations $\Sigma_{\pm}$, which refer to the integral paths along the gray area and the white area indicated by arrows in Fig.\ref{fig:contour}. Applying the Cauchy projectors to the equation \eqref{4-regular-Jump}, there is
\begin{equation*}
	{\bf M}^{[1]}_{\pm}(\lambda;x,t)=-\frac{1}{2\pi\ii}\int_{\Sigma_{\pm}}\dfrac{{\bf M}_{-}(\zeta;x,t){\bf J}(\zeta;x,t)}{\zeta-(\lambda\pm\ii 0)}\dd\zeta.
\end{equation*}
Therefore,
\begin{equation}\label{4-M-explicit-infity}
	\begin{split}
	&{\bf M}_{\pm}(\lambda;x,t)
	=\exp\left(\frac{\ii\sigma_{3}}{2}\int_{x}^{+\infty}|q(y,t)|^{2}dy\right)-\frac{1}{2\pi\ii}\int_{\Sigma_{\pm}}\dfrac{{\bf M}_{-}(\zeta;x,t){\bf J}(\zeta;x,t)}{\zeta-(\lambda\pm\ii 0)}\dd\zeta\\[4pt]
	+&\sum_{k=1}^{2N}\left[\dfrac{c_{k} \psi_{2}^{+}(\lambda_{k};x,t) }{\lambda-\lambda_{k}}\exp\left(2\ii\big(\lambda_{k}^{2}x+2\mathcal{F}_{fd}(\lambda_{k}^{2})t\big)\right),\ \dfrac{d_{k}\psi_{1}^{+}(\lambda_{k}^{*};x,t)}{\lambda-\lambda_{k}^{*}}\exp\left(-2\ii\big(\lambda_{k}^{*2}x+2\mathcal{F}_{fd}(\lambda_{k}^{*2})t\big)\right)\right],
\end{split}
\end{equation}
where $c_{k}{=}c_{k}(\lambda_{k};t){=}\frac{a_{k}}{s_{11}^{'}(\lambda_{k};t)},\ d_{k}{=} d_{k}(\lambda_{k}^{*};t){=}\frac{b_{k}}{s_{22}^{'}(\lambda_{k}^{*};t)}$. Combining with the symmetries of ${\bf S}(\lambda;t)$ and the relations of $a_{k}$ and $b_{k}$, we can deduce
\begin{equation*}
	d_{k}(\lambda_{k}^{*};t)=-c^{*}_{k}(\lambda_{k};t)=-c^{*}_{k}(-\lambda_{k};t),\ \ \ \ c_{k+N}(\lambda_{k+N};t)=c_{k}(\lambda_{k};t).
\end{equation*}

Then we can recover the potential function $q(x,t)$ from ${\bf M}_{+}(\lambda;x,t)$. Firstly, we expand ${\bf M}_{+}(\lambda;x,t)$ at large-$\lambda$ as
\begin{equation}\label{4-M-expand}
	{\bf M}_{+}(\lambda;x,t)={\bf M}_{+,0}(x,t)+\dfrac{{\bf M}_{+,1}(x,t)}{\lambda}+\mathcal{O}(\lambda^{-2}),
\end{equation}
and we know 
\begin{equation*}
	\begin{split}
	&{\bf M}_{+,0}(x,t)=\exp\left(\frac{\ii\sigma_{3}}{2}\int_{x}^{+\infty}|q(y,t)|^{2}dy\right),\\
	&{\bf M}_{+,1}(x,t)=\frac{1}{2\pi\ii}\int_{\Sigma_{+}}{\bf M}_{-}(\zeta;x,t){\bf J}(\zeta;x,t)\dd\zeta+\sum_{k=1}^{2N}\Big[c_{k}\psi_{2}^{+}(\lambda_{k};x,t)\exp\left(2\ii\left(\lambda_{k}^{2}x+2\mathcal{F}_{fd}(\lambda_{k}^{2})t\right)\right),\\
	&\ \ \ \ \ \ \ \ \ \ \ \ \ \ \ \ \ \ \ \ \ \ \ \ \ \ \ \ \ \ \ \ \ \ \ \ \ \ \ \ \ \ \ \ \ \ \ \ \ \ \ \ \ \ \ \ \ \ \ \ \ \ \ \ \ \ \ \ \ \ \ \ \ \ \ \ \ \ \ \ \ \ \ \ \ \ \ \ \ \  d_{k}\psi_{1}^{+}(\lambda_{k}^{*};x,t)\exp\left(-2\ii\left(\lambda_{k}^{*2}x+2\mathcal{F}_{fd}(\lambda_{k}^{*2})t\right)\right)\Big].
	\end{split}
\end{equation*}
Then the potential function $q(x,t)$ can be recovered by substituting the expansion of ${\bf M}_{+}(\lambda;x,t)$ (i.e.\eqref{4-M-expand}) into \eqref{4-Psi}, and collecting the same powers of $\lambda$. We derive
\begin{equation*}
	{\bf Q}(x,t)=\ii\big[\sigma_{3},{\bf M}_{+,1}(x,t)\big]\ ({\bf M}_{+,0}(x,t))^{-1}.
\end{equation*}
Therefore,
\begin{equation*}
	\begin{split}
	q(x,t)=\exp\left(\frac{\ii}{2}\int_{x}^{+\infty}|q(y,t)|^{2}dy\right)\bigg(\frac{1}{\pi}\int_{\Sigma_{+}}&\Big({\bf M}_{-}(\zeta;x,t){\bf J}(\zeta;x,t)\Big)_{12}\dd\zeta\\
	&+2\ii\sum_{k=1}^{2N}d_{k}\psi_{11}^{+}(\lambda_{k}^{*};x,t)\exp\left(-2\ii\big(\lambda_{k}^{*2}x+2\mathcal{F}_{fd}(\lambda_{k}^{*2})t\big)\right)\bigg).
\end{split}
\end{equation*}

Next we will try to give the explicit expressions for $s_{11}(\lambda;t)$ and $s_{22}(\lambda;t)$ by constructing a new analytic function $\beta(\lambda;t)$:
\begin{equation*}
	\begin{split}
		&\beta_{+}(\lambda;t)=s_{11}(\lambda;t)\exp\left(\frac{\ii}{2}\int_{-\infty}^{+\infty}|q(y,t)|^{2}dy\right)\prod_{k=1}^{2N}\dfrac{\lambda-\lambda_{k}^{*}}{\lambda-\lambda_{k}},\ \ \ \ \ \ \lambda\in D^{+},\\
		&\beta_{-}(\lambda;t)=s_{22}(\lambda;t)\exp\left(-\frac{\ii}{2}\int_{-\infty}^{+\infty}|q(y,t)|^{2}dy\right)\prod_{k=1}^{2N}\dfrac{\lambda-\lambda_{k}}{\lambda-\lambda_{k}^{*}},\ \ \ \ \lambda\in D^{-}.
	\end{split}
\end{equation*}
Obviously, $\beta_{\pm}(\lambda;t)$ are analytic in $D^{\pm}$, and $\beta_{\pm}(\lambda;t)\to1$ as $\lambda\to\infty$. In addition, we have the relation:
\begin{equation*}
	\log\beta_{+}(\lambda;t)+\log\beta_{-}(\lambda;t)=-\log\Big(1-\rho_{1}(\lambda;t)\rho_{2}(\lambda;t)\Big),\ \ \ \ \lambda\in\Sigma.
\end{equation*}
Applying the Cauchy projectors to  the above equation, we can get
\begin{equation*}
	\log\beta_{\pm}(\lambda;t)=\mp\frac{1}{2\pi\ii}\int_{\Sigma_{\pm}}\dfrac{\log\Big(1-\rho_{1}(\zeta;t)\rho_{2}(\zeta;t)\Big)}{\zeta-(\lambda\pm\ii 0)}\dd\zeta.
\end{equation*}
Then the expressions of $s_{11}(\lambda;t)$ and $s_{22}(\lambda;t)$ are as follows:
\begin{equation*}
	\begin{split}
		&s_{11}(\lambda;t)=\exp\Bigg(-\frac{1}{2\pi\ii}\int_{\Sigma_{+}}\dfrac{\log\Big(1-\rho_{1}(\zeta;t)\rho_{2}(\zeta;t)\Big)}{\zeta-(\lambda+\ii 0)}\dd\zeta-\frac{\ii}{2}\int_{-\infty}^{+\infty}|q(y,t)|^{2}dy\Bigg)\prod_{k=1}^{2N}\dfrac{\lambda-\lambda_{k}}{\lambda-\lambda_{k}^{*}},\\[4pt]
		&s_{22}(\lambda;t)=\exp\Bigg(\frac{1}{2\pi\ii}\int_{\Sigma_{-}}\dfrac{\log\Big(1-\rho_{1}(\zeta;t)\rho_{2}(\zeta;t)\Big)}{\zeta-(\lambda-\ii 0)}\dd\zeta+\frac{\ii}{2}\int_{-\infty}^{+\infty}|q(y,t)|^{2}dy\Bigg)\prod_{k=1}^{2N}\dfrac{\lambda-\lambda_{k}^{*}}{\lambda-\lambda_{k}}.
	\end{split}
\end{equation*}

\subsection{Explicit form of the fDNLS equation}\label{4-subsection-fDNLS}
In order to find the explicit form of the fDNLS equation, the first question to consider is how the recursion operator function $\mathcal{F}_{fd}(\mathcal{L})$ acts on functions. Note that squared eigenfunctions are eigenfunctions of the recursion operator of integrable equations. So we can let the recursion operator $\mathcal{L}$ act on squared eigenfunctions, and then generalize this to the case of recursion operator function $\mathcal{F}_{fd}(\mathcal{L})$. Then we need to consider how to connect squared eigenfunctions with potential functions, which can be achieved through the completeness of squared eigenfunctions. The completeness of squared eigenfunctions of the fDNLS equation is closely related to the perturbation theory (or variational relations) \cite{yang-2010}, which can be found in detail in the literatures \cite{kawata1980generalized,kawata1980linear}, and we will be briefly described below.

Firstly, we introduce squared eigenfunctions ${\bf\Omega}_{j}(\lambda;x,t)=\begin{bmatrix}
	\phi_{1j}^{2}(\lambda;x,t),&
	\phi_{2j}^{2}(\lambda;x,t)
\end{bmatrix}^{\top}$, where
\begin{equation*}
	{\bf\Omega}_{+,1}=\begin{bmatrix}
		(\phi_{11}^{-})^{2}\\[4pt]
		(\phi_{21}^{-})^{2}
	\end{bmatrix},\ \ \ \ 
	{\bf\Omega}_{+,2}=\begin{bmatrix}
		(\phi_{12}^{+})^{2}\\[4pt]
		(\phi_{22}^{+})^{2}
	\end{bmatrix},\ \ \ \ 
	{\bf\Omega}_{-,1}=\begin{bmatrix}
		(\phi_{11}^{+})^{2}\\[4pt]
		(\phi_{21}^{+})^{2}
	\end{bmatrix},\ \ \ \ 
	{\bf\Omega}_{-,2}=\begin{bmatrix}
		(\phi_{12}^{-})^{2}\\[4pt]
		(\phi_{22}^{-})^{2}
	\end{bmatrix},
\end{equation*}
the subscripts $_{\pm}$ indicate that the functions ${\bf\Omega}_{\pm,j}(\lambda;x,t)$ are analytic in $D^{\pm}$, respectively. And we can calculate that ${\bf\Omega}_{j}(\lambda;x,t)$ satisfies 
\begin{equation}\label{4-Omega-prop}
		{\bf\Omega}_{j,x}(\lambda;x,t)+2\ii\lambda^{2}\sigma_{3}{\bf\Omega}_{j}(\lambda;x,t)=-\lambda\big\langle\phi_{j}(\lambda;x,t)\big|\sigma_{3}\big|\phi_{j}(\lambda;x,t)\big\rangle\sigma_{1}{\bf w}(x,t),
\end{equation}
where ${\bf w}(x,t)=\begin{bmatrix}
	-q^{*},&q
\end{bmatrix}^{\top}$, the ``ket" $|{\bf v}\rangle$ represents a usual column vector ${\bf v}=|{\bf v}\rangle=\begin{bmatrix}
v_{1},&v_{2}
\end{bmatrix}^{\top}$, and the ``bra" $\langle{\bf v}|=\begin{bmatrix}
-v_{2},&v_{1}
\end{bmatrix}$ is its adjoint raw vector. Through a direct calculation, we can derive $\mathcal{L}(\sigma_{3}\partial_{x}{\bf\Omega}_{j})=-\lambda^{2}\sigma_{3}\partial_{x}{\bf\Omega}_{j},$ 
here the recursion operator $\mathcal{L}$ corresponds to the DNLS equation. Then we generalize it to
\begin{equation*}
	\mathcal{F}_{fd}(\mathcal{L})(\sigma_{3}\partial_{x}{\bf\Omega}_{j})=\mathcal{F}_{fd}(-\lambda^{2})\sigma_{3}\partial_{x}{\bf\Omega}_{j}.
\end{equation*}

Next, we will deduce the variational relations of the fDNLS equation. Let us consider a perturbation $\Delta{\bf Q}(x,t)$ for the first equation of \eqref{4-Lax-Phi}, then the corresponding variation of  $\Delta\Phi^{\pm}(\lambda;x,t)$ can be described as
\begin{equation*}
	\Delta\Phi^{\pm}_{x}(\lambda;x,t)={\bf U}(\lambda;x,t)\Delta\Phi^{\pm}(\lambda;x,t)+\lambda\Delta{\bf Q}(x,t)\Phi^{\pm}(\lambda;x,t),\ \ \ \ \ \Delta\Phi^{\pm}(\lambda;x,t)\to 0,\ \ {\rm as\ }x\to\pm\infty.
\end{equation*}
The above equation can be easily solved
\begin{equation}\label{4-delta-Phi}
	\Delta\Phi^{\pm}(\lambda;x,t)=\Phi^{\pm}(\lambda;x,t)\int_{\pm\infty}^{x}\lambda\big(\Phi^{\pm}(\lambda;y,t)\big)^{-1}\Delta{\bf Q}(y,t)\Phi^{\pm}(\lambda;y,t)\dd y.
\end{equation}
On the other hand, considering the variation of \eqref{4-Phi-S}, we can derive
\begin{equation}\label{4-var-S}
	\Delta{\bf S}(\lambda;t)=\left(\Phi^{+}(\lambda;x,t)^{-1}\left(\Delta\Phi^{-}(\lambda;x,t)-\Delta\Phi^{+}(\lambda;x,t){\bf S}(\lambda;t)\right)\right).
\end{equation}
Then substituting \eqref{4-delta-Phi} into \eqref{4-var-S}, and combining the relation \eqref{4-Phi-S}, 
\begin{equation*}
	\Delta{\bf S}(\lambda;t)=\int_{-\infty}^{\infty}\lambda\big(\Phi^{+}(\lambda;y,t)\big)^{-1}\Delta{\bf Q}(y,t)\Phi^{-}(\lambda;y,t)\dd y.
\end{equation*}
Combining with the above equation, we can obtain
\begin{equation}\label{4-var-rho}
	\begin{split}
		&\Delta\rho_{1}(\lambda;t)=\frac{1}{s_{11}^{2}}\big(s_{11}\Delta s_{21}-s_{21}\Delta s_{11}\big)=\dfrac{\lambda}{s_{11}^{2}(\lambda;t)} \int_{-\infty}^{+\infty}\big\langle{\bf\Omega}_{+,1}(\lambda;y,t)\big|\sigma_{1}\big|\Delta{\bf w}(y,t)\big\rangle\dd y,\\[2pt]
		&\Delta\rho_{2}(\lambda;t)=\frac{1}{s_{22}^{2}}\big(s_{22}\Delta s_{12}-s_{12}\Delta s_{22}\big)=-\dfrac{\lambda}{s_{22}^{2}(\lambda;t)} \int_{-\infty}^{+\infty}\big\langle{\bf\Omega}_{-,2}(\lambda;y,t)\big|\sigma_{1}\big|\Delta{\bf w}(y,t)\big\rangle\dd y.
	\end{split}
\end{equation}
Equation \eqref{4-var-rho} can be regarded as a mapping from $\Delta q(x,t)$ to $\Delta\rho_{j}(\lambda;t)$. And then we wish to construct its inverse mapping by discussing two integrals. 

The first integral is
\begin{equation}\label{4-integral}
	{\bf E}_{n}(x,t)=\frac{1}{2\pi\ii}\int_{\mathbb{E}}\lambda^{n}\dd\lambda\int_{-\infty}^{\infty}{\bf G}(\lambda;x,y,t){\bf v}(y,t)\dd y,\ \ \ \ n=0,1,2,
\end{equation}
where $\mathbb{E}$ is a contour path enclosing the whole region of the
$\lambda$-plane, ${\bf v}(x,t)$ is an arbitrary smooth vector function, ${\bf G}(\lambda;x,y,t)$ is called the Green function which will be given bellow. As shown in Fig.\ref{fig:contour}, the path $\mathbb{E}$ can be divided into two half-circular paths $\Gamma_{\pm}$ or two contour paths $\Gamma_{\pm}+\Sigma_{\pm}$, i.e. $\mathbb{E}=\Gamma_{+}-\Gamma_{-}=(\Gamma_{+}+\Sigma_{+})-(\Gamma_{-}+\Sigma_{-})$. The Green function ${\bf G}(\lambda;x,y,t)$ is defined as follows:
\begin{equation*}
	{\bf G}_{x}(\lambda;x,y,t)-{\bf U}(\lambda;x,t){\bf G}(\lambda;x,y,t)=\delta(x-y)\mathbb{I},
\end{equation*}
where $\delta(z)$ is a Dirac's $\delta$-function, and we choose the following two kinds of Green functions, whose detailed construction procedure can be found in Appendix-A in \cite{kawata1980generalized},
\begin{equation*}\label{4-green function}
	\begin{split}
		&{\bf G}_{p}(\lambda;x,y,t)=
		\begin{cases}
			\ \ \big|\phi_{2}^{+}(\lambda;x,t)\big\rangle\dfrac{1}{s_{11}(\lambda;t)}\big\langle\phi_{1}^{-}(\lambda;y,t)\big|,\ \ \ \ y\leq x,\\[10pt]
			\ \ 	\big|\phi_{1}^{-}(\lambda;x,t)\big\rangle\dfrac{1}{s_{11}(\lambda;t)}\big\langle\phi_{2}^{+}(\lambda;y,t)\big|,\ \ \ \ y\geq x,
		\end{cases}\\
		&{\bf G}_{n}(\lambda;x,y,t)=
		\begin{cases}
			-\big|\phi_{1}^{+}(\lambda;x,t)\big\rangle\dfrac{1}{s_{22}(\lambda;t)}\big\langle\phi_{2}^{-}(\lambda;y,t)\big|,\ \ \ \ y\leq x,\\[10pt]
			-\big|\phi_{2}^{-}(\lambda;x,t)\big\rangle\dfrac{1}{s_{22}(\lambda;t)}\big\langle\phi_{1}^{+}(\lambda;y,t)\big|,\ \ \ \ y\geq x,
		\end{cases}
	\end{split}
\end{equation*}
where ${\bf G}_{p}(\lambda;x,y,t)$ and ${\bf G}_{n}(\lambda;x,y,t)$ are defined on $D^{\pm}$, respectively. By calculation, we get
\begin{equation}\label{4-integral-relation}
	{\bf E}_{0}(x,t)=0,\ \ \ \ {\bf E}_{1}(x,t)=-\ii\sigma_{3}{\bf v}(x,t),\ \ \ \ {\bf E}_{2}(x,t)={\bf Q}(x,t){\bf v}(x,t).
\end{equation}
And based on the relations \eqref{4-Phi-S} and \eqref{4-rho}, we can rewrite the integral \eqref{4-integral} as
\begin{multline}\label{4-reintegral}
	{\bf E}_{n}(x,t)=\frac{1}{2\pi\ii}\int_{-\infty}^{\infty}\bigg(\int_{\Gamma_{-}+\Sigma_{-}}\big|\phi_{1}^{+}(\lambda;x,t)\big\rangle\rho_{2}(\lambda;t)\big\langle\phi_{1}^{+}(\lambda;y,t)\big|\lambda^{n}\dd\lambda\\[4pt]
	+\int_{\Gamma_{+}+\Sigma_{+}}\big|\phi_{2}^{+}(\lambda;x,t)\big\rangle\rho_{1}(\lambda;t)\big\langle\phi_{2}^{+}(\lambda;y,t)\big|\lambda^{n}\dd\lambda-\int_{\Sigma_{+}}\big|\phi^{+}(\lambda;x,t)\big\rangle\rho(\lambda;t)\big\langle\phi^{+}(\lambda;y,t)\big|\lambda^{n}\dd\lambda\bigg){\bf v}(y,t)\dd y,	
\end{multline}
where
\begin{multline}\label{4-phi-rho}
	\big|\phi^{+}(\lambda;x,t)\big\rangle\rho(\lambda;t)\big\langle\phi^{+}(\lambda;y,t)\big|=\big|\phi_{1}^{+}(\lambda;x,t)\big\rangle\big\langle\phi_{2}^{+}(\lambda;y,t)\big|+\big|\phi_{2}^{+}(\lambda;x,t)\big\rangle\big\langle\phi_{1}^{+}(\lambda;y,t)\big|\\[3pt]
	+\big|\phi_{1}^{+}(\lambda;x,t)\big\rangle\rho_{2}(\lambda;t)\big\langle\phi_{1}^{+}(\lambda;y,t)\big|+\big|\phi_{2}^{+}(\lambda;x,t)\big\rangle\rho_{1}(\lambda;t)\big\langle\phi_{2}^{+}(\lambda;y,t)\big|.
\end{multline}
Combining with equations \eqref{4-integral-relation} and \eqref{4-reintegral}, we can get the completeness of the fundamental solutions
\begin{equation}\label{4-F-spec}
	{\bf F}_{0}(x,y,t)=0,\ \ \ \ {\bf F}_{1}(x,y,t)=\sigma_{3}\delta(x-y),\ \ \ \ {\bf F}_{2}(x,y,t)=\ii{\bf Q}(x,t)\delta(x-y),
\end{equation}
where
\begin{multline}\label{4-F}
	{\bf F}_{n}(x,y,t)=\frac{1}{2\pi}\bigg(\int_{\Gamma_{-}+\Sigma_{-}}\big|\phi_{1}^{+}(\lambda;x,t)\big\rangle\rho_{2}(\lambda;t)\big\langle\phi_{1}^{+}(\lambda;y,t)\big|\lambda^{n}\dd\lambda\\[4pt]
	+\int_{\Gamma_{+}+\Sigma_{+}}\big|\phi_{2}^{+}(\lambda;x,t)\big\rangle\rho_{1}(\lambda;t)\big\langle\phi_{2}^{+}(\lambda;y,t)\big|\lambda^{n}\dd\lambda
	-\int_{\Sigma_{+}}\big|\phi^{+}(\lambda;x,t)\big\rangle\rho(\lambda;t)\big\langle\phi^{+}(\lambda;y,t)\big|\lambda^{n}\dd\lambda\bigg).
\end{multline}

The second integral is
\begin{equation}\label{4-int}
	-\frac{1}{2\pi}\int_{\Sigma_{+}}\big|\widetilde{\phi}^{+}(\lambda;x,t)\big\rangle\widetilde{\rho}(\lambda;t)\big\langle\phi^{+}(\lambda;y,t)\big|\lambda^{n}\dd\lambda,\ \ \ \ n=0,1,
\end{equation}
where $\big|\widetilde{\phi}^{+}(\lambda;x,t)\big\rangle\widetilde{\rho}(\lambda;t)\big\langle\phi^{+}(\lambda;y,t)\big|$ is defined in the same way as \eqref{4-phi-rho}. To facilitate the discussion of the above integral, we need to introduce the integral representations of  $\Phi(\lambda;x,t)$ and $\widetilde{\Phi}(\lambda;x,t)$, which are the fundamental solutions of the Lax pairs corresponding to the potential functions ${\bf Q}(x,t)$ and $\widetilde{{\bf Q}}(x,t)={\bf Q}+\Delta{\bf Q}$, respectively. By combining these integral representations and their inverse forms, we can get 
\begin{equation}\label{4-jost-trig-1}
	\big|\widetilde{\phi}_{j}^{+}(\lambda;x,t)\big\rangle={\bf K}(x,t)\big|\phi_{j}^{+}(\lambda;x,t)\big\rangle-\int_{x}^{+\infty}\big({\bf L}_{d}(x,y,t)+\lambda{\bf L}_{o}(x,y,t)	\big)\big|\phi_{j}^{+}(\lambda;y,t)\big\rangle\dd y,\ \ j=1,2,
\end{equation}
\begin{equation}\label{4-jost-trig-2}
	\big\langle\phi_{j}^{+}(\lambda;x,t)\big|=\big\langle\widetilde{\phi}_{j}^{+}(\lambda;x,t)\big|{\bf \widetilde K}^{\rm A}(x,t)-\int_{x}^{+\infty}\big\langle\widetilde{\phi}_{j}^{+}(\lambda;y,t)\big|\big({\bf \widetilde L}_{d}^{\rm A}(x,y,t)+\lambda{\bf \widetilde L}_{o}^{\rm A}(x,y,t)\big)\dd y,\ \ j=1,2,
\end{equation}
where ${\bf K}$ and ${\bf L}_{d}$ are diagonal, ${\bf L}_{o}$ is off-diagonal, the superscript $^{\rm A}$ refers to the adjoint matrix, and 
\begin{equation}\label{4-K}
	{\bf K}(x,t)=\exp\left(\frac{\ii}{2}\int_{x}^{+\infty}\big(|\tilde{q}(y,t)|^{2}-|q(y,t)|^{2}\big)\dd y\sigma_{3}\right),\ \ 2{\bf L}_{o}(x,x,t)=\widetilde{\bf Q}(x,t){\bf K}(x,t)-{\bf K}(x,t){\bf Q}(x,t).
\end{equation}
For the case of ${\bf \widetilde Q}(x,t)$, we distinguish all related quantities by marking a ``tilde". Then we discuss the integral \eqref{4-int} when $n=1$ from two aspects, the first is to substitute \eqref{4-jost-trig-1} into \eqref{4-int} for simplification, the second is to substitute \eqref{4-jost-trig-2} into \eqref{4-int}, and then combining with \eqref{4-F-spec}, \eqref{4-F}, \eqref{4-jost-trig-1} for simplification. By comparing the integrals obtained by different simplification methods, we can get
\begin{multline}\label{4-KFW}
	\int_{x}^{+\infty}\Big({\bf L}_{d}(x,z,t)\big({\bf F}_{1}(z,y,t)+{\bf W}_{1}(z,y,t)\big)+{\bf L}_{o}(x,z,t)\big({\bf F}_{2}(z,y,t)+{\bf W}_{2}(z,y,t)\big)\Big)\dd z\\
	={\bf K}(x,t)\big({\bf F}_{1}(x,y,t)+{\bf W}_{1}(x,y,t)\big),\ \ \ \ x\neq y,
\end{multline}
where $\Delta\rho(\lambda;t)=\widetilde{\rho}(\lambda;t)-\rho(\lambda;t)$,
\begin{equation}\label{4-def-Wn}
	\begin{split}
	{\bf W}_{n}(x,y,t)=\frac{1}{2\pi}\Big(\int_{\Gamma_{+}}\big|\phi_{2}^{+}(\lambda;x,t)\big\rangle\Delta\rho_{1}(\lambda;t)&\big\langle\phi^{+}_{2}(\lambda;y,t)\big|\lambda^{n}\dd\lambda\\
	&+\int_{\Gamma_{-}}\big|\phi_{1}^{+}(\lambda;x,t)\big\rangle\Delta\rho_{2}(\lambda;t)\big\langle\phi^{+}_{1}(\lambda;y,t)\big|\lambda^{n}\dd\lambda\Big).
\end{split}
\end{equation}
According to the symmetry properties of $\Phi(\lambda;x,t)$ and ${\bf S}(\lambda;t)$ in proposition \ref{prop-sym}, we can find that ${\bf W}_{1}$ is diagonal, ${\bf W}_{0}$ and ${\bf W}_{2}$ are off-diagonal. Based on \eqref{4-F-spec}, the equation \eqref{4-KFW} can be rewritten as
\begin{multline*}
	{\bf K}(x,t){\bf W}_{1}(x,y,t)-\int_{x}^{+\infty}\Big({\bf L}_{d}(x,z,t){\bf W}_{1}(z,y,t)+{\bf L}_{o}(x,z,t){\bf W}_{2}(z,y,t)\Big)\dd z\\
	-{\bf L}_{d}(x,y,t)\sigma_{3}-\ii{\bf L}_{o}(x,y,t){\bf Q}(y,t)=0,\ \ \ \ x\neq y,
\end{multline*}
which is the generalized Gel'fand-Levitan (G-L) equation. Similarly, considering the integral \eqref{4-int} when $n=0$, we can deduce another generalized G-L equation:
\begin{equation*}
	{\bf K}(x,t){\bf W}_{0}(x,y,t)-\int_{x}^{+\infty}\Big({\bf L}_{d}(x,z,t){\bf W}_{0}(z,y,t)+{\bf L}_{o}(x,z,t){\bf W}_{1}(z,y,t)\Big)\dd z-{\bf L}_{o}(x,y,t)\sigma_{3}=0,\ \ x\neq y.
\end{equation*}

Based on the above preparations, we will construct the mapping from $\Delta\rho_{j}(\lambda;t)$ to $\Delta q(x,t)$. We use the notation $\delta\gamma$, a quantity related to $\gamma$, to represent $ \parallel\delta\gamma\parallel$ is sufficiently small. Note that for sufficiently small $\|\Delta\rho_{j}(\lambda;t)\|,\ j=1,2$, the generalized G-L equations can be reduced to the linear equations,
\begin{equation}\label{4-linear}
{\bf L}_{d}(x,y,t)\sigma_{3}={\bf K}(x,t){\bf W}_{1}(x,y,t)-\ii{\bf L}_{o}(x,y,t){\bf Q}(y,t),\ \ \ \ {\bf L}_{o}(x,y,t)\sigma_{3}={\bf K}(x,t){\bf W}_{0}(x,y,t).
\end{equation}
According to \eqref{4-K}, \eqref{4-linear} and the definition of ${\bf W}_{n}(x,y,t)$ (i.e.\eqref{4-def-Wn}), we can get two different expressions for $\begin{bmatrix}
	\big({\bf W}_{0}(x,x,t)\big)_{12},&
	\big({\bf W}_{0}(x,x,t)\big)_{21}
\end{bmatrix}^{\top}$. Comparing these two expressions, we can derive
\begin{equation}\label{4-map}
	\begin{split}
	\sigma_{1}\delta {\bf w}(x,t)=
	&-\frac{1}{\pi}\left(\int_{\Gamma_{+}}{\bf\Omega}_{+,2}(\lambda;x,t)\delta\rho_{1}(\lambda;t)\dd\lambda+\int_{\Gamma_{-}}{\bf\Omega}_{-,1}(\lambda;x,t)\delta\rho_{2}(\lambda;t)\dd\lambda\right)\\
	&+\ii\int_{x}^{+\infty}\big\langle {\bf w}(y,t)\big|\sigma_{3}\big|\delta{\bf w}(y,t)\big\rangle\dd y\ \sigma_{3}\sigma_{1}{\bf w}(x,t).
	\end{split}
\end{equation}
A direct calculation yields
\begin{equation}\label{4-relation}
	\begin{split}
		&\big\langle{\phi_{1}^{+}(\lambda;x,t)|\sigma_{3}|\phi_{1}^{+}(\lambda;x,t)}\big\rangle=-2\lambda\int_{x}^{+\infty}\big\langle{\bf w}(y,t)|\sigma_{1}\sigma_{3}|{\bf\Omega}_{-,1}(\lambda;y,t)\big\rangle\dd y,\\
		&\big\langle{\phi_{2}^{+}(\lambda;x,t)|\sigma_{3}|\phi_{2}^{+}(\lambda;x,t)}\big\rangle=-2\lambda\int_{x}^{+\infty}\big\langle{\bf w}(y,t)|\sigma_{1}\sigma_{3}|{\bf\Omega}_{+,2}(\lambda;y,t)\big\rangle\dd y.
	\end{split}
\end{equation}
Based on \eqref{4-relation}, the equation \eqref{4-map} can be rewritten as:
\begin{multline}\label{4-delta-alpha}
	\int_{x}^{+\infty}\big\langle {\bf w}(y,t)\big|\sigma_{3}\big|\delta{\bf w}(y,t)\big\rangle\dd y=\frac{1}{2\pi}\bigg(\int_{\Gamma_{+}}\frac{\delta\rho_{1}(\lambda;t)}{\lambda}
	\big\langle{\phi_{2}^{+}(\lambda;x,t)|\sigma_{3}|\phi_{2}^{+}(\lambda;x,t)}\big\rangle
	\dd\lambda\\
	+\int_{\Gamma_{-}}\frac{\delta\rho_{2}(\lambda;t)}{\lambda}\big\langle{\phi_{1}^{+}(\lambda;x,t)|\sigma_{3}|\phi_{1}^{+}(\lambda;x,t)}\big\rangle\dd\lambda\bigg).
\end{multline}
Substituting \eqref{4-Omega-prop} and \eqref{4-delta-alpha} into \eqref{4-map} gives
\begin{equation}\label{4-var-w}
	\delta{\bf w}(x,t)=-\dfrac{\sigma_{2}}{2\pi}\left(\int_{\Gamma_{+}}\dfrac{\delta\rho_{1}(\lambda;t)}{\lambda^{2}}\ \partial_{x}{\bf\Omega}_{+,2}(\lambda;x,t)\dd\lambda+\int_{\Gamma_{-}}\dfrac{\delta\rho_{2}(\lambda;t)}{\lambda^{2}}\ \partial_{x}{\bf\Omega}_{-,1}(\lambda
	;x,t)\dd\lambda\right).
\end{equation}
Based on the time evolution of ${\bf S}(\lambda;t)$, the time evolution of $\delta\rho_{j}(\lambda;t),\ j=1,2$ are found to be
\begin{equation}\label{4-evolution-deltarho}
	\delta\rho_{1}(\lambda;t)=\ee^{-4\ii\mathcal{F}_{fd}(\lambda^{2})t}\delta\rho_{j}(\lambda;0),\ \ \ \ \delta\rho_{2}(\lambda;t)=\ee^{4\ii\mathcal{F}_{fd}(\lambda^{2})t}\delta\rho_{2}(\lambda;0).
\end{equation} 
Combining with \eqref{4-evolution-deltarho} and substituting \eqref{4-var-rho} into \eqref{4-var-w},
\begin{multline}\label{4-map-w-rho}
	\delta{\bf w}(x,t)=-\dfrac{\sigma_{2}}{2\pi}\bigg(\int_{\Gamma_{+}}\dfrac{1}{\lambda s_{11}^{2}(\lambda;t)}\ \partial_{x}{\bf\Omega}_{+,2}(\lambda;x,t)\int_{-\infty}^{+\infty}\big\langle{\bf\Omega}_{+,1}(\lambda;y,t)\big|\sigma_{1}\big|\delta{\bf w}(y,t)\big\rangle\dd y\dd\lambda\\
	-\int_{\Gamma_{-}}\dfrac{1}{\lambda s_{22}^{2}(\lambda;t)}\  \partial_{x}{\bf\Omega}_{-,1}(\lambda
	;x,t)\int_{-\infty}^{+\infty}\big\langle{\bf\Omega}_{-,2}(\lambda;y,t)\big|\sigma_{1}\big|\delta{\bf w}(y,t)\big\rangle\dd y\dd\lambda\bigg).
\end{multline}
Furthermore, 
\begin{equation}\label{4-w-delta}
	\delta{\bf w}(x,t)=\int_{-\infty}^{+\infty}\delta(x-y)\delta{\bf w}(y,t)\dd y.
\end{equation}
Then we can get a completeness relation of squared eigenfunctions by comparing \eqref{4-map-w-rho} with \eqref{4-w-delta},
\begin{multline*}
	\sigma_{3}\delta(x-y)=\frac{1}{2\pi\ii}\bigg(\int_{\Gamma_{+}}\dfrac{1}{\lambda s_{11}^{2}(\lambda;t)}\left(\ \partial_{x}\big|{\bf\Omega}_{+,2}(\lambda;x,t)\big\rangle\right)\big\langle{\bf\Omega}_{+,1}(\lambda;y,t)\big|\dd\lambda\\
	-\int_{\Gamma_{-}}\dfrac{1}{\lambda s_{22}^{2}(\lambda;t)}\left(\ \partial_{x}\big|{\bf\Omega}_{-,1}(\lambda;x,t)\big\rangle\right)\big\langle{\bf\Omega}_{-,2}(\lambda;y,t)\big|\dd\lambda\bigg).
\end{multline*}
Based on the equation \eqref{4-map-w-rho}, we can assume a sufficiently smooth and decaying vector function $\vartheta(x,t)=\begin{bmatrix}
	\vartheta_{1},&\vartheta_{2}
\end{bmatrix}^{\top}$, which can also be expanded in terms of the eigenfunctions. We choose $\vartheta(x,t)=\begin{bmatrix}
q^{*},&q
\end{bmatrix}^{\top}$, and let $\mathcal{F}_{fd}(\mathcal{L})$ act on it, then there is
\begin{multline}\label{4-Fv}
	\mathcal{F}_{fd}(\mathcal{L})\begin{bmatrix}
		q^{*}\\
		q
	\end{bmatrix}=-\frac{\sigma_{2}}{2\pi}\bigg(\int_{\Gamma_{+}}\dfrac{\lambda^{3}|2\lambda^{2}|^{\epsilon}}{s_{11}^{2}(\lambda;t)}\begin{bmatrix}
		(\phi_{12}^{+})^2\\[3pt]	(\phi_{22}^{+})^2
	\end{bmatrix}_{x}
	\int_{-\infty}^{+\infty}\Big(\big(\phi_{11}^{-}\big)^{2}q^{*}-\big(\phi_{21}^{-}\big)^{2}q\Big)\dd y\dd\lambda\\[4pt]
	-\int_{\Gamma_{-}}\dfrac{\lambda^{3}|2\lambda^{2}|^{\epsilon}}{s_{22}^{2}(\lambda;t)}\begin{bmatrix}
		(\phi_{11}^{+})^2\\[3pt]	(\phi_{21}^{+})^2
	\end{bmatrix}_{x}
	\int_{-\infty}^{+\infty}\Big(\big(\phi_{12}^{-}\big)^{2}q^{*}-\big(\phi_{22}^{-}\big)^{2}q\Big)\dd y\dd\lambda\bigg),
\end{multline}
in the integral terms of \eqref{4-Fv}, $q{=}q(y,t),\ q^{*}{=}q^{*}(y,t),\ \phi_{jk}^{+}{=}\phi_{jk}^{+}(\lambda;x,t),\  \phi_{jk}^{-}{=}\phi_{jk}^{-}(\lambda;y,t)$,\ $j,k{=}1,2.$ Then the explicit form of the fDNLS equation can be given by combining with the equation \eqref{4-evolution eq},
	\begin{equation}\label{4-fDNLS-new}
	\begin{split}
		q_{t}(x,t)=&-\frac{2}{\pi}\bigg(\int_{\Gamma_{+}}\dfrac{\lambda^{3}|2\lambda^{2}|^{\epsilon}}{ s_{11}^{2}(\lambda;t)}\partial_{x}(\phi_{12}^{+}(\lambda;x,t))^2
		\int_{-\infty}^{+\infty}\Big(\big(\phi_{11}^{-}(\lambda;y,t)\big)^{2}q^{*}(y,t)-\big(\phi_{21}^{-}(\lambda;y,t)\big)^{2}q(y,t)\Big)\dd y\dd\lambda\\
		&-\int_{\Gamma_{-}}\dfrac{\lambda^{3}|2\lambda^{2}|^{\epsilon}}{ s_{22}^{2}(\lambda;t)}\partial_{x}(\phi_{11}^{+}(\lambda;x,t))^2
		\int_{-\infty}^{+\infty}\Big(\big(\phi_{12}^{-}(\lambda;y,t)\big)^{2}q^{*}(y,t)-\big(\phi_{22}^{-}(\lambda;y,t)\big)^{2}q(y,t)\Big)\dd y\dd\lambda\bigg),
	\end{split}
\end{equation}
in the integral of the above equation, $q=q(y,t),\ q^{*}=q^{*}(y,t),\ \phi_{1k}^{+}=\phi_{1k}^{+}(\lambda;x,t),\  \phi_{jk}^{-}=\phi_{jk}^{-}(\lambda;y,t),\ j,k=1,2.$ 
In particular, the equation \eqref{4-fDNLS-new} will degenerate into the classical DNLS equation when $\epsilon=0$.

\section{Fractional $N$-soliton solution}
In this section, we want to explore the fractional $N$-soliton solution of the fDNLS equation, which leads us to start with the case of reflectionless potential: $\rho_{1}(\lambda;t)=\rho_{2}(\lambda;t)=0$. Then there is ${\bf J}(\lambda;x,t)=0$, so 
\begin{equation*}
	s_{11}(\lambda;t){=}\exp\left({-}\frac{\ii}{2}\int_{-\infty}^{+\infty}|q(y,t)|^{2}dy\right)\prod_{k=1}^{2N}\dfrac{\lambda-\lambda_{k}}{\lambda-\lambda_{k}^{*}},\ \ 
	s_{22}(\lambda;t){=}\exp\left(\frac{\ii}{2}\int_{-\infty}^{+\infty}|q(y,t)|^{2}dy\right)\prod_{k=1}^{2N}\dfrac{\lambda-\lambda_{k}^{*}}{\lambda-\lambda_{k}},
\end{equation*}
and
\begin{equation}\label{4-q-reflectionless}
	q(x,t)=2\ii\exp\left(\frac{\ii}{2}\int_{x}^{+\infty}|q(y,t)|^{2}dy\right)\sum_{k=1}^{2N}d_{k}\psi_{11}^{+}(\lambda_{k}^{*};x,t)\exp\left(-2\ii\big(\lambda_{k}^{*2}x+2\mathcal{F}_{fd}(\lambda_{k}^{*2})t\big)\right).
\end{equation}
Since the equation \eqref{4-q-reflectionless} contains the unknown function $\psi_{11}^{+}(\lambda_{k}^{*};x,t)$, so we will look for the expression for this function. Based on the definitions of ${\bf M}_{\pm}(\lambda;x,t)$ (i.e.\eqref{4-M-def}) and their explicit forms (i.e.\eqref{4-M-explicit-infity}), we can derive
\begin{equation}\label{4-psi_11}
	\psi_{11}^{+}(\lambda;x,t)=\exp\left(\frac{\ii}{2}\int_{x}^{+\infty}|q(y,t)|^{2}dy\right)+\sum_{k=1}^{2N}\dfrac{c_{k}}{\lambda-\lambda_{k}}\exp\left(2\ii\big(\lambda_{k}^{2}x+2\mathcal{F}_{fd}(\lambda_{k}^{2})t\big)\right)\psi_{12}^{+}(\lambda_{k};x,t),
\end{equation}
\begin{equation}\label{4-psi_12}
	\psi_{12}^{+}(\lambda;x,t)=\sum_{k=1}^{2N}\dfrac{d_{k}}{\lambda-\lambda_{k}^{*}}\exp\left(-2\ii\big(\lambda_{k}^{*2}x+2\mathcal{F}_{fd}(\lambda_{k}^{*2})t\big)\right)\psi_{11}^{+}(\lambda_{k}^{*};x,t).
\end{equation}
Taking $\lambda=\lambda_{k}$ in \eqref{4-psi_12}, and substituting it into \eqref{4-psi_11}. Then taking $\lambda=\lambda_{k}^{*}$,
\begin{multline*}
	\psi_{11}^{+}(\lambda_{k}^{*};x,t)=\exp\left(\frac{\ii}{2}\int_{x}^{+\infty}|q(y,t)|^{2}dy\right)-\sum_{j=1}^{2N}\sum_{l=1}^{2N}\bigg(\dfrac{c_{l}d_{j}}{(\lambda_{l}-\lambda_{k}^{*})(\lambda_{l}-\lambda_{j}^{*})}\psi_{11}^{+}(\lambda_{j}^{*};x,t)\\
	\times\exp\Big(2\ii\big((\lambda_{l}^{2}-\lambda_{j}^{*2})x+2\big(\mathcal{F}_{fd}(\lambda_{l}^{2})-\mathcal{F}_{fd}(\lambda_{j}^{*2})\big)t\big)\Big)\bigg).
\end{multline*}
Using the method as in \cite{zhang2020derivative}, the solution to the above equation can be expressed as follows:
\begin{equation}\label{4-psi11-infty}
	\psi_{11}^{+}(\lambda_{k}^{*};x,t)=\exp\left(\frac{\ii}{2}\int_{x}^{+\infty}|q(y,t)|^{2}dy\right)\dfrac{\det\widetilde{{\bf R}}_{k}}{\det{\bf R}},
\end{equation}
where ${\bf R}=\mathbb{I}+\sum\limits_{l=1}^{2N}{\bf R}_{0,l}$, the element at the position $(j,k),\ j,k=1,\cdots,2N$ of the matrix ${\bf R}_{0,l}$ is 
\begin{equation*}
	({\bf R}_{0,l})_{j,k}=\dfrac{c_{l}d_{k}}{(\lambda_{l}-\lambda_{j}^{*})(\lambda_{l}-\lambda_{k}^{*})}\exp\Big(2\ii\big((\lambda_{l}^{2}-\lambda_{k}^{*2})x+2\big(\mathcal{F}_{fd}(\lambda_{l}^{2})-\mathcal{F}_{fd}(\lambda_{k}^{*2})\big)t\big)\Big),
\end{equation*}
$\widetilde{{\bf R}}_{k}$ is replacing the $k$-th column of the matrix ${\bf R}$ with the column vector ${\bf e}=[\ \underbrace{1,1,\cdots,1}_{2N}\ ]^{\top}$. Based on the equation \eqref{4-psi11-infty}, the function $q(x,t)$ in \eqref{4-q-reflectionless} becomes
\begin{equation*}
	q(x,t)=2\ii\exp\left(\ii\int_{x}^{+\infty}|q(y,t)|^{2}dy\right)\sum_{k=1}^{2N}d_{k}\exp\left(-2\ii\big(\lambda_{k}^{*2}x+2\mathcal{F}_{fd}(\lambda_{k}^{*2})t\big)\right)\dfrac{\det\widetilde{{\bf R}}_{k}}{\det{\bf R}}.
\end{equation*}
We denote $d_{k}\exp\left(-2\ii\big(\lambda_{k}^{*2}x+2\mathcal{F}_{fd}(\lambda_{k}^{*2})t\big)\right)=\widetilde{d}_{k}$, then the above equation can be rewritten as 
\begin{equation}\label{4-q-lambda-infty}
	q(x,t)=-2\ii\exp\left(\ii\int_{x}^{+\infty}|q(y,t)|^{2}dy\right)\dfrac{\det{\bf R}^{{\rm e}}}{\det{\bf R}},\ \ \ \ \ \ 
	{\bf R}^{{\rm e}}=\begin{bmatrix}
		0&\widetilde{{\bf d}}\\
		{\bf e}&{\bf R}
	\end{bmatrix},\ \ \ \ \widetilde{{\bf d}}=\begin{bmatrix}
		\widetilde{d}_{1},&\widetilde{d}_{2},&\cdots,&\widetilde{d}_{2N}
	\end{bmatrix}.
\end{equation}
Obviously, the equation \eqref{4-q-lambda-infty} is an implicit one, and we need further analysis to get an explicit form of $q(x,t)$. Note that the above analyses were discussed as $\lambda\to\infty$, and we can also consider the expansion of $\lambda\to 0$ by using the same method. Combining with the idea in \cite{zhang2020derivative}, we can get another representation of ${\bf M}(\lambda;x,t)$,
\begin{equation}\label{4-M-explicit-0}
	\begin{split}
		&{\bf M}(\lambda;x,t)=\mathbb{I}
		+\lambda\sum_{k=1}^{2N}\left(\dfrac{\underset{\lambda=\lambda_{k}}{\rm{Res}}\frac{{\bf M}(\lambda;x,t)}{\lambda}}{\lambda-\lambda_{k}}+\dfrac{\underset{\lambda=\lambda_{k}^{*}}{\rm{Res}}\frac{{\bf M}(\lambda;x,t)}{\lambda}}{\lambda-\lambda_{k}^{*}}\right)\\[4pt]
		&=\mathbb{I}
		{+}\lambda\sum_{k=1}^{2N}\left[\dfrac{c_{k} \psi_{2}^{+}(\lambda_{k};x,t) }{\lambda_{k}(\lambda-\lambda_{k})} \exp\left(2\ii\left(\lambda_{k}^{2}x{+}2\mathcal{F}_{fd}(\lambda_{k}^{2})t\right)\right) ,\ \dfrac{d_{k}\psi_{1}^{+}(\lambda_{k}^{*};x,t)}{\lambda_{k}^{*}(\lambda-\lambda_{k}^{*})} \exp\left({-}2\ii\left(\lambda_{k}^{*2}x{+}2\mathcal{F}_{fd}(\lambda_{k}^{*2})t\right)\right)\right].
	\end{split}
\end{equation}
Now we based on the definitions of ${\bf M}_{\pm}(\lambda;x,t)$ (i.e.\eqref{4-M-def}) and their explicit forms (i.e.\eqref{4-M-explicit-0}) to reconsider the explicit form of $\psi_{11}^{+}(\lambda_{k}^{*};x,t)$. Similarly, we can obtain another equation related to $\psi_{11}^{+}(\lambda_{k}^{*};x,t)$
\begin{multline*}
	\psi_{11}^{+}(\lambda_{k}^{*};x,t)=1-\lambda_{k}^{*}\sum_{j=1}^{2N}\sum_{l=1}^{2N}\bigg(\dfrac{c_{l}d_{j}}{(\lambda_{l}-\lambda_{k}^{*})(\lambda_{l}-\lambda_{j}^{*})\lambda_{j}^{*}}\psi_{11}^{+}(\lambda_{j}^{*};x,t)\\
	\times\exp\Big(2\ii\big((\lambda_{l}^{2}-\lambda_{j}^{*2})x+2\big(\mathcal{F}_{fd}(\lambda_{l}^{2})-\mathcal{F}_{fd}(\lambda_{j}^{*2})\big)t\big)\Big)\bigg),
\end{multline*} 
which can be solved explicitly by
\begin{equation}\label{4-psi11-0}
	\psi_{11}^{+}(\lambda_{k}^{*};x,t)=\dfrac{\det\widetilde{{\bf T}}_{k}}{\det{\bf T}},
\end{equation}
where ${\bf T}=\mathbb{I}+\sum\limits_{l=1}^{2N}{\bf T}_{0,l}$, the $(j,k)$-element of the matrix ${\bf T}_{0,l}$ is given by
\begin{equation*}
	({\bf T}_{0,l})_{j,k}=\dfrac{c_{l}d_{k}\lambda_{j}^{*}}{(\lambda_{l}-\lambda_{j}^{*})(\lambda_{l}-\lambda_{k}^{*})\lambda_{k}^{*}}\exp\Big(2\ii\big((\lambda_{l}^{2}-\lambda_{k}^{*2})x+2\left(\mathcal{F}_{fd}(\lambda_{l}^{2})-\mathcal{F}_{fd}(\lambda_{k}^{*2})\right)t\big)\Big),
\end{equation*}
and $\widetilde{{\bf T}}_{k}$ is the matrix ${\bf T}$ by replacing the $k$-th column with the column vector ${\bf e}$. Substituting \eqref{4-psi11-0} into \eqref{4-q-reflectionless},
	\begin{equation*}
		q(x,t)=2\ii\exp\left(\frac{\ii}{2}\int_{x}^{+\infty}|q(y,t)|^{2}dy\right)\sum_{k=1}^{2N}d_{k}\exp\left(-2\ii\big(\lambda_{k}^{*2}x+2\mathcal{F}_{fd}(\lambda_{k}^{*2})t\big)\right)\dfrac{\det\widetilde{{\bf T}}_{k}}{\det{\bf T}}.
	\end{equation*}
	Similarly, we denote $d_{k}\exp\left(-2\ii\big(\lambda_{k}^{*2}x+2\mathcal{F}_{fd}(\lambda_{k}^{*2})t\big)\right)=\widetilde{d}_{k}$, then the above equation can be rewritten as
\begin{equation}\label{4-q-lambda-0}
	q(x,t)=-2\ii\exp\left(\frac{\ii}{2}\int_{x}^{+\infty}|q(y,t)|^{2}dy\right)\dfrac{\det{\bf T}^{{\rm e}}}{\det{\bf T}},\ \ \ \ {\bf T}^{{\rm e}}=\begin{bmatrix}
		0&\widetilde{{\bf d}}\\
		{\bf e}&{\bf T}
	\end{bmatrix}.
\end{equation}
Comparing \eqref{4-q-lambda-infty} with \eqref{4-q-lambda-0}, there is
\begin{equation}\label{4-exp}
	\exp\left(\frac{\ii}{2}\int_{x}^{+\infty}|q(y,t)|^{2}dy\right)=\dfrac{\det{\bf T}^{{\rm e}}\det{\bf R}}{\det{\bf R}^{{\rm e}}\det{\bf T}}.
\end{equation}
Substituting \eqref{4-exp} into the equation \eqref{4-q-lambda-infty} or \eqref{4-q-lambda-0}, then the explicit form of the fractional $N$-soliton solution can be derived
\begin{equation}\label{4-q-soliton}
	q(x,t)=-2\ii\left(\dfrac{\det{\bf T}^{{\rm e}}}{\det{\bf T}}\right)^{2}\dfrac{\det{\bf R}}{\det{\bf R}^{{\rm e}}}.
\end{equation}

When $N=1$, we choose the spectral parameter $\lambda_{1}=\xi+\ii\eta$, which implies $\lambda_{2}=-\xi-\ii\eta.$ Then the fractional one-soliton solution can be obtained according to the equation \eqref{4-q-soliton}, which is summarized in the following proposition. And we will prove this fractional one-soliton solution solves the fDNLS equation.
\begin{proposition}
	The expression of the fractional one-soliton solution of the fDNLS equation \eqref{4-fDNLS-new} is as follows:
	\begin{equation}\label{4-q1}
		q^{[1]}(x,t)=-2\ii c_{1}\exp\left(\omega_{1}+\theta^{*}-2\theta\right)\cosh\left(\omega_{2}+\theta^{*}\right){\rm sech}^{2}\left(\omega_{2}+\theta\right),
	\end{equation}
	where $\ee^{\theta}=\dfrac{c_{1}(\xi-\ii\eta)}{2\xi\eta}$,
	\begin{equation*}
		\begin{split}
			&	\omega_{1}(x,t)=-2\ii\left((\xi^{2}-\eta^{2})x+2^{1+\epsilon}(\xi^{4}+\eta^{4}-6\xi^{2}\eta^{2})(\xi^{2}+\eta^{2})^{\epsilon}t\right),\\
			&\omega_{2}(x,t)=-4\xi\eta\left(x+2^{2+\epsilon}(\xi^{2}-\eta^{2})(\xi^{2}+\eta^{2})^{\epsilon}t\right).
		\end{split}
	\end{equation*}
\end{proposition}
\begin{proof}
	For convenience, we assume $\theta=\theta_{R}+\ii\theta_{I}$. By observing the equation \eqref{4-fDNLS-new}, it can be found that we first need to give the functions $\phi^{\pm}_{jk}(\lambda;x,t),\ j,k=1,2$, which correspond to the solution $q^{[1]}(x,t)$. After careful consideration, we believe that it is most convenient to find the expression
	of $\phi^{\pm}_{jk}(\lambda;x,t)$ by using the Darboux transform method. By modifying the form of the Darboux matrix in \cite{guo2013high}  properly, we obtain the one-fold Darboux matrix corresponds to the Lax pair of \eqref{4-fDNLS-new}
	\begin{equation*}
		{\bf D}=\mathbb{I}-\lambda\lambda_{1}^{*}\left(\dfrac{{\bf A}}{\lambda-\lambda_{1}^{*}}+\dfrac{\sigma_{3}{\bf A}\sigma_{3}}{\lambda+\lambda_{1}^{*}}\right),
	\end{equation*}
	where
	\begin{equation*}
		\begin{split}
			&{\bf A}=\dfrac{\lambda_{1}^{2}-\lambda_{1}^{*2}}{2|\lambda_{1}|^{2}}\begin{bmatrix}
				\alpha_{1}&0\\[3pt]
				0&\alpha_{2}
			\end{bmatrix}\varphi_{1}\varphi_{1}^{\dagger},\ \ \ \ \alpha_{1}^{-1}=\varphi_{1}^{\dagger}\begin{bmatrix}
				\lambda_{1}&0\\[3pt]
				0&\lambda_{1}^{*}
			\end{bmatrix}\varphi_{1},\ \ \ \ \alpha_{2}^{-1}=\varphi_{1}^{\dagger}\begin{bmatrix}
				\lambda_{1}^{*}&0\\[3pt]
				0&\lambda_{1}
			\end{bmatrix}\varphi_{1},\\[5pt]
			&\varphi_{1}=\begin{bmatrix}
				-\ii\exp\left(-\ii\lambda_{1}^{2}\big(x+2\lambda_{1}^{2}|2\lambda_{1}^{2}|^{\epsilon}t\big)-\frac{\theta_{R}+3\ii\theta_{I}}{2}\right)\\[5pt]
				\exp\left(\ii\lambda_{1}^{2}\big(x+2\lambda_{1}^{2}|2\lambda_{1}^{2}|^{\epsilon}t\big)+\frac{\theta_{R}+3\ii\theta_{I}}{2}\right)
			\end{bmatrix},
		\end{split}
	\end{equation*}
	the superscript $^\dagger$ denotes the complex conjugation and
	vector transpose. 
	And the transformation between the potential function $q(x,t)$ and the new potential function $q^{[1]}(x,t)$ is
	\begin{equation}\label{4-q-soliton-DT}
		q^{[1]}=q+({\bf A}-\sigma_{3}{\bf A}\sigma_{3})_{12,x}.
	\end{equation}
	Note that the fractional one-soliton solution $q^{[1]}(x,t)$ obtained by Darboux transform method (i.e.\eqref{4-q-soliton-DT} is consistent with the solution obtained by IST (i.e.\eqref{4-q1}). Then the new eigenfunctions $\Phi^{\pm[1]}(\lambda;x,t):=\widehat{\Phi}^{\pm}(\lambda;x,t)$ which correspond to the solution $q^{[1]}(x,t)$ can be derived by applying the asymptotic behavior \eqref{4-Phi-asy-be} and the relation $\widehat{\Phi}^{\pm}={\bf D}\Phi^{\pm}_{0}$,
	\begin{equation}\label{4-q1-Phi}
		\begin{split}
			&\widehat{\Phi}=\dfrac{(\xi{-}\ii\eta)}{(\xi+\ii\eta)(\lambda+\xi-\ii\eta)(\lambda-\xi+\ii\eta)}\\[8pt]
			&\ \ \ \times
			\begin{bmatrix}
				\left(\lambda^{2}\cosh(\omega_{2}{+}\theta^{*}){\rm sech}(\omega_{2}{+}\theta){-}(\xi^{2}{+}\eta^{2})\right)\ee^{\delta}&c_{1}\lambda\exp(\omega_{1}{-}\theta_{R}{-}3\ii\theta_{I}){\rm sech}(\omega_{2}{+}\theta)\ee^{{-}\delta}\\[10pt]
				c_{1}\lambda\exp({-}\omega_{1}{-}\theta_{R}{+}3\ii\theta_{I}){\rm sech}(\omega_{2}{+}\theta^{*})\ee^{\delta}&\left(\lambda^{2}\cosh(\omega_{2}{+}\theta){\rm sech}(\omega_{2}{+}\theta^{*}){-}(\xi^{2}{+}\eta^{2})\right)\ee^{{-}\delta}
			\end{bmatrix},\\[4pt]
			&\widehat{\Phi}^{+}{=}\widehat{\Phi}\begin{bmatrix}
				\dfrac{(\xi+\ii\eta)^{2}(\lambda+\xi-\ii\eta)(\lambda-\xi+\ii\eta)}{(\xi-\ii\eta)^{2}(\lambda-\xi-\ii\eta)(\lambda+\xi+\ii\eta)}&0\\[10pt]
				0&1
			\end{bmatrix},\  \widehat{\Phi}^{-}{=}\widehat{\Phi}\begin{bmatrix}
				1&0\\[10pt]
				0&\dfrac{(\xi+\ii\eta)^{2}(\lambda+\xi-\ii\eta)(\lambda-\xi+\ii\eta)}{(\xi-\ii\eta)^{2}(\lambda-\xi-\ii\eta)(\lambda+\xi+\ii\eta)}
			\end{bmatrix},
		\end{split}
	\end{equation}
	where $\Phi_{0}(\lambda;x,t)={\rm diag}(\ee^{\delta},\ee^{-\delta})$, $\delta=-\ii\lambda^{2}(x+2\lambda^{2}|2\lambda^{2}|^{\epsilon}t)$.
	
	Next, we will prove that $q^{[1]}(x,t)$ (i.e.\eqref{4-q1}) is a solution of the equation \eqref{4-fDNLS-new}. Firstly, we want to calculate the integral part of the equation \eqref{4-fDNLS-new} on $\Gamma_{+}$, i.e.,
	\begin{equation}\label{4-step1}
		\begin{split}
		\int_{\Gamma_{+}}\dfrac{\lambda^{3}|2\lambda^{2}|^{\epsilon}}{ s_{11}^{2}(\lambda;t)}\partial_{x}\left(\widehat{\phi}_{12}^{+}(\lambda;x,t)\right)^2
		\int_{-\infty}^{+\infty}\left(\left(\widehat{\phi}_{11}^{-}(\lambda;y,t)\right)^{2}\left(q^{[1]}(y,t)\right)^{*}-\left(\widehat{\phi}_{21}^{-}(\lambda;y,t)\right)^{2}q^{[1]}(y,t)\right)\dd y\dd\lambda&\\:=\int_{\Gamma_{+}}\dfrac{g(\lambda;x,t)}{s_{11}^{2}(\lambda;t)}\dd\lambda&,
		\end{split}
	\end{equation}
	where $g(\lambda;x,t)=g_{1}(\lambda;x,t)\int_{-\infty}^{+\infty}g_{2}(\lambda;y,t)\dd y,$
	\begin{equation*}
			g_{1}(\lambda;x,t)=\lambda^{3}|2\lambda^{2}|^{\epsilon}\partial_{x}\left(\widehat{\phi}_{12}^{+}(\lambda;x,t)\right)^2,\ \ \ \ g_{2}(\lambda;y,t)=\left(\widehat{\phi}_{11}^{-}(\lambda;y,t)\right)^{2}\left(q^{[1]}(y,t)\right)^{*}-\left(\widehat{\phi}_{21}^{-}(\lambda;y,t)\right)^{2}q^{[1]}(y,t).
	\end{equation*}
	Combining with the residue theorem, we can decompose \eqref{4-step1} into continuous and discrete parts. The function $g(\lambda;x,t)$ is analytic in $D^{+}$, so the residues in \eqref{4-step1} come from the zeros of $s_{11}(\lambda;t)$, which occur at $\lambda_{1}=\xi+\ii\eta$ and $\lambda_{2}=-\xi-\ii\eta$. Then we have 
	\begin{equation}\label{4-step1-cd}
		\begin{split}
			\int_{\Gamma_{+}}\dfrac{g(\lambda;x,t)}{s_{11}^{2}(\lambda;t)}\dd\lambda=&-\int_{\Sigma_{+}}\dfrac{g(\lambda;x,t)}{s_{11}^{2}(\lambda;t)}\dd\lambda+2\pi\ii\sum_{j=1}^{2}\underset{\lambda=\lambda_{j}}{\rm{Res}}\dfrac{g(\lambda;x,t)}{s_{11}^{2}(\lambda;t)}\\
			=&-\int_{\Sigma_{+}}\dfrac{g(\lambda;x,t)}{s_{11}^{2}(\lambda;t)}\dd\lambda+2\pi\ii\sum_{j=1}^{2}\left(\dfrac{g^{'}(\lambda_{j};x,t)}{s_{11}^{'2}(\lambda_{j};t)}-\dfrac{g(\lambda_{j};x,t)s_{11}^{''}(\lambda_{j})}{s_{11}^{'3}(\lambda_{j};t)}\right),
		\end{split}
	\end{equation}
	here the superscript $^{'}$ also denotes the partial derivative with respect to $\lambda$. 
	Through some analyses and calculations, we find that the above equation only needs to calculate the first part of the discrete spectrum, which comes from $\int_{-\infty}^{+\infty}g_{2}(\lambda;y,t)\dd y=0$ for all $\lambda$. In fact, based on \eqref{4-q1} and \eqref{4-q1-Phi},
	\begin{equation*}
		\begin{split}
		g_{2}(\lambda;y,t)=&\dfrac{2\ii c_{1}(\xi-\ii\eta)^2\exp(-\theta_{R}+3\ii\theta_{I}-\omega_{1}+2\delta)}{\big((\xi+\ii\eta)(\lambda+\xi-\ii\eta)(\lambda-\xi+\ii\eta)\big)^2}\Big(\lambda^{4}{\rm sech}(\omega_{2}+\theta)
		-2\lambda^{2}(\xi^2+\eta^2){\rm sech}(\omega_{2}+\theta^{*})\\
		&+(\xi^2+\eta^2)^2\cosh(\omega_{2}+\theta){\rm sech}^2(\omega_{2}+\theta^{*})+c_{1}^2\lambda^2\ee^{-2\theta_{R}}{\rm sech}(\omega_{2}+\theta^{*}){\rm sech}^2(\omega_{2}+\theta)\Big).
		\end{split}
	\end{equation*}
	We introduce the transformation $z=\omega_{2}+\theta_{R}$, then the above equation can be rewritten as 
	\begin{equation*}
		\begin{split}
		g_{2}(\lambda;z,t)=&\dfrac{2\ii c_{1}(\xi-\ii\eta)^2\exp(-\theta_{R}+3\ii\theta_{I}+\widehat{\delta})}{\big((\xi+\ii\eta)(\lambda+\xi-\ii\eta)(\lambda-\xi+\ii\eta)\big)^2}\Big(\lambda^{4}{\rm sech}(z+\ii\theta_{I})
		-2\lambda^{2}(\xi^2+\eta^2){\rm sech}(z-\ii\theta_{I})\\
		&+(\xi^2+\eta^2)^2\cosh(z+\ii\theta_{I}){\rm sech}^2(z-\ii\theta_{I})+c_{1}^2\lambda^2\ee^{-2\theta_{R}}{\rm sech}(z-\ii\theta_{I}){\rm sech}^2(z+\ii\theta_{I})\Big),\\
		\widehat{\delta}=&\ \ii 2^{2+\epsilon}\left(2\lambda^2(\xi^{2}-\eta^{2})(\xi^{2}+\eta^{2})^{\epsilon}-(\xi^2+\eta^2)^{2+\epsilon}-|\lambda^2|^{\epsilon}\lambda^4\right)t+\dfrac{\ii(\xi^2-\eta^2-\lambda^2)\theta_{R}}{2\xi\eta}.
	\end{split}
	\end{equation*}
	Therefore,
	\begin{equation*}
		\begin{split}
		\int_{-\infty}^{+\infty}g_{2}(\lambda;y,t)\dd y=&\dfrac{-\ii c_{1}(\xi-\ii\eta)^2\exp(-\theta_{R}+3\ii\theta_{I}+\widehat{\delta})}{2\xi\eta\big((\xi+\ii\eta)(\lambda+\xi-\ii\eta)(\lambda-\xi+\ii\eta)\big)^2}\int_{-\infty}^{+\infty}\widehat{g}_{2}(\lambda;z,t)\dd z,\\
		\widehat{g}_{2}(\lambda;z,t)=&\exp\left(\dfrac{\ii(\lambda^2-\xi^2+\eta^2)z}{2\xi\eta}\right)\Big(\lambda^{4}{\rm sech}(z+\ii\theta_{I})
		-2\lambda^{2}(\xi^2+\eta^2){\rm sech}(z-\ii\theta_{I})\\
		&+(\xi^2+\eta^2)^2\cosh(z+\ii\theta_{I}){\rm sech}^2(z-\ii\theta_{I})+c_{1}^2\lambda^2\ee^{-2\theta_{R}}{\rm sech}(z-\ii\theta_{I}){\rm sech}^2(z+\ii\theta_{I})\Big).
		\end{split}
	\end{equation*}
	We consider $\widehat{g}_{2}(\lambda;z,t)$ on the matrix contour in Fig.\ref{fig:z-plane}, there is
	\begin{figure}
		\centering
		\includegraphics[width=0.45\linewidth]{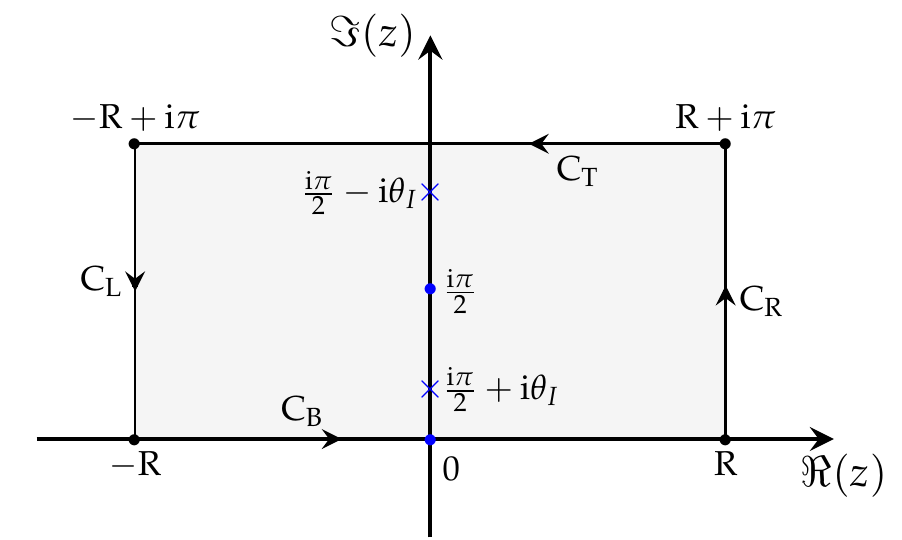}
		\caption[z-plane]{z-plane}
		\label{fig:z-plane}
	\end{figure}
	\begin{equation*}
		\int_{C_{B}}\widehat{g}_{2}(z)\dd z+\int_{C_{R}}\widehat{g}_{2}(z)\dd z+\int_{C_{T}}\widehat{g}_{2}(z)\dd z+\int_{C_{L}}\widehat{g}_{2}(z)\dd z
		=2\pi\ii\left(\underset{z=\frac{\ii\pi}{2}+\ii\theta_{I}}{\rm{Res}}\widehat{g}_{2}(z)+\underset{z=\frac{\ii\pi}{2}-\ii\theta_{I}}{\rm{Res}}\widehat{g}_{2}(z)\right).
	\end{equation*}
	By observation, we can find the integral over $C_{L}$ and $C_{R}$ vanish as $R\to\infty$, and the integral over $C_{T}$ can be written in terms of the integral over $C_{B}$ as
	\begin{equation*}
		\int_{C_{T}}\widehat{g}_{2}(z)\dd z=\exp\left(\dfrac{-(\lambda^2-\xi^2+\eta^2)\pi}{2\xi\eta}\right)\int_{C_{B}}\widehat{g}_{2}(z)\dd z,
	\end{equation*}
	then
	\begin{equation*}
		\left(1+\exp\left(\dfrac{-(\lambda^2-\xi^2+\eta^2)\pi}{2\xi\eta}\right)\right)\int_{C_{B}}\widehat{g}_{2}(z)\dd z=2\pi\ii\left(\underset{z=\frac{\ii\pi}{2}+\ii\theta_{I}}{\rm{Res}}\widehat{g}_{2}(z)+\underset{z=\frac{\ii\pi}{2}-\ii\theta_{I}}{\rm{Res}}\widehat{g}_{2}(z)\right).
	\end{equation*}
	While the residues of $\widehat{g}_{2}(z)$ vanish at $z=\frac{\ii\pi}{2}\pm\ii\theta_{I}$. Thus, $\int_{C_{B}}\widehat{g}_{2}(z)\dd z=0$, which leads to $\int_{-\infty}^{+\infty}g_{2}(\lambda;y,t)\dd y=0$. Combining with the expression for $s_{11}(\lambda;t)$, we can rewrite \eqref{4-step1-cd} as 
	\begin{equation*}
		\int_{\Gamma_{+}}\dfrac{g(\lambda;x,t)}{s_{11}^{2}(\lambda;t)}\dd\lambda=2\pi\ii\sum_{j=1}^{2}\dfrac{g'(\lambda_{j};x,t)}{s_{11}^{'2}(\lambda_{j};t)}=\dfrac{-16\ii\pi\xi^2\eta^2(\xi+\ii\eta)^2}{(\xi-\ii\eta)^4}g'(\lambda_{1};x,t).
	\end{equation*}
	So we only need to calculate $g'(\lambda_{1};x,t)$, and in fact, $g'(\lambda_{1};x,t)=g_{1}(\lambda_{1};x,t)\int_{-\infty}^{\infty}g^{'}_{2}(\lambda_{1};x,t)\dd y$. Through  calculations, we can get
	\begin{equation*}
		g'(\lambda_{1};x,t)=\dfrac{\ii(\xi^2+\eta^2)^2(\xi-\ii\eta)^2\ee^{\omega_{1}-3\ii\theta_{I}}}{4\xi^3\eta^3\big((\xi-\ii\eta)\ee^{\omega_{2}+\theta_{R}}+(\xi+\ii\eta)\ee^{-\omega_{2}-\theta_{R}}\big)}
		\Big(c_{1}^2(\xi^2+\eta^2)^2\ee^{2\omega_{2}}+4\xi^2\eta^2(\xi^2-6\ii\xi\eta-\eta^2)\Big).
	\end{equation*}
	Therefore,
	\begin{equation}\label{4-int-gamma+}
		\int_{\Gamma_{+}}\dfrac{g(\lambda;x,t)}{s_{11}^{2}(\lambda;t)}\dd\lambda{=}2^{\epsilon}(\xi^2{+}\eta^2)^{\epsilon}{\rm sech}^3(\omega_{2}{+}\theta)\ee^{\omega_{1}}\left(\dfrac{\pi c_{1}^2(\xi{+}\ii\eta)^6}{2\xi\eta(\xi{-}\ii\eta)}\ee^{2\omega_{2}}{+}\dfrac{2\pi\xi\eta(\xi{+}\ii\eta)^4(\xi^2{-}6\ii\xi\eta{-}\eta^2)}{(\xi{-}\ii\eta)^3}\right).
	\end{equation}
	
	Using the same method as in calculating the integral on $\Gamma_{+}$, we can also calculate  the integral part of the equation \eqref{4-fDNLS-new} on $\Gamma_{-}$, and here we give the result directly,
	\begin{equation}\label{4-int-gamma-}
		\begin{split}
		&\int_{\Gamma_{-}}\dfrac{\lambda^{3}|2\lambda^{2}|^{\epsilon}}{ s_{22}^{2}(\lambda;t)}\partial_{x}\left(\widehat{\phi}_{11}^{+}(\lambda;x,t)\right)^2
		\int_{-\infty}^{+\infty}\left(\left(\widehat{\phi}_{12}^{-}(\lambda;y,t)\right)^{2}\left(q^{[1]}(y,t)\right)^{*}-\left(\widehat{\phi}_{22}^{-}(\lambda;y,t)\right)^{2}q^{[1]}(y,t)\right)\dd y\dd\lambda\\
		&=-2^{\epsilon}(\xi^2+\eta^2)^{\epsilon}{\rm sech}^3(\omega_{2}+\theta)\ee^{\omega_{1}}
		\times\left(\dfrac{8\pi\xi^3\eta^3(\xi-\ii\eta)}{c_{1}^2}\ee^{-2\omega_{2}}+2\pi\xi\eta(\xi-\ii\eta)(\xi^2+6\ii\xi\eta-\eta^2)\right).
		\end{split}
	\end{equation}
	
	In addition, according to the equation \eqref{4-q1}, we can directly obtain the derivative of $q^{[1]}(x,t)$ with respect to $t$,
	\begin{equation}\label{4-q1-t}
		\begin{split}
		q^{[1]}_{t}(x,t)=-2^{\epsilon}(\xi^2+\eta^2)^{\epsilon}{\rm sech}^3(\omega_{2}+\theta)\ee^{\omega_{1}}\bigg(\dfrac{c_{1}^2(\xi+\ii\eta)^6}{\xi\eta(\xi-\ii\eta)}\ee^{2\omega_{2}}&+\dfrac{16\xi^3\eta^3(\xi-\ii\eta)}{c_{1}^2}\ee^{-2\omega_{2}}\\&+\dfrac{8\xi\eta(\xi^2-\eta^2)(\xi^4+18\xi^2\eta^2+\eta^4)}{(\xi-\ii\eta)^3}\bigg).
		\end{split}
	\end{equation}
	By substituting \eqref{4-int-gamma+},  \eqref{4-int-gamma-}, and \eqref{4-q1-t} into the equation \eqref{4-fDNLS-new}, it can be found that the left and right sides are equal. This proves that $q^{[1]}(x,t)$ (i.e.\eqref{4-q1}) satisfies the equation \eqref{4-fDNLS-new}.
\end{proof}

Next, we will analyze the properties of the solution $q^{[1]}(x,t)$. In terms of \eqref{4-q1}, we can easily obtain the expression for the modular square of the solution $q^{[1]}(x,t)$,
\begin{equation*}
	|q^{[1]}(x,t)|^2=\dfrac{32\xi^2\eta^2}{(\xi^2+\eta^2)\cosh(2\omega_{2}+2\theta_{R})+\xi^2-\eta^2}.
\end{equation*}
Furthermore, we can derive the maximum value of $|q^{[1]}(x,t)|$ is $4\eta$. By selecting appropriate parameters, we give the relevant figures corresponding to $|q^{[1]}(x,t)|$ in Fig.\ref{fig:direction}. 
From the left column in Fig.\ref{fig:direction}, we can find that the
soliton is a left-going traveling-wave soliton. And with the
increase of $\epsilon$, the velocity of the traveling wave will be faster, which can be observed from the right column in Fig.\ref{fig:direction}.
\begin{figure}[H]
	\centering
	\includegraphics[width=0.6\linewidth]{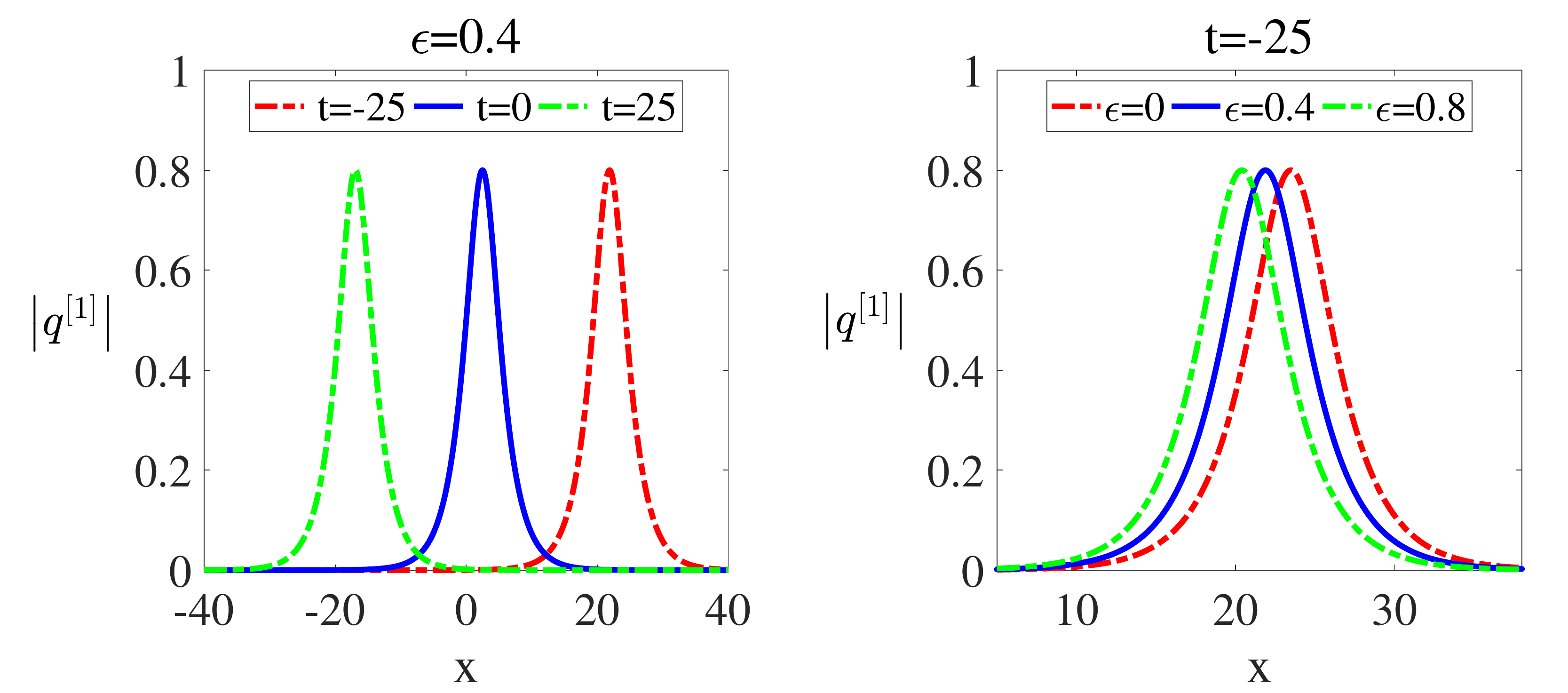}
	\caption[wave direction]{The direction of wave propagation. Choosing the parameters: $\xi=0.5,\eta=0.2,c_{1}=1.$}
	\label{fig:direction}
\end{figure}

In addition, if we take the limit of the soliton solution, we can also obtain the rational solution \cite{xu2011darboux,guo2013high}. By considering $\xi\to0$ in the fractional one-soliton solution \eqref{4-q1}, we find that the fractional rational solution $q^{[1]}_{r}(x,t)$ will occur when $c_{1}=\pm2\xi\to0$. Then we have
\begin{equation}\label{4-q1-rational}
	q^{[1]}_{r}(x,t)=\mp4\ii\eta\exp\left(2\ii\eta^2\left(x-(2\eta^2)^{1+\epsilon}t\right)\right)\dfrac{4\eta^2\left(x-2(2\eta^2)^{1+\epsilon}t\right)+\ii}{\left(4\eta^2\left(x-2(2\eta^2)^{1+\epsilon}t\right)-\ii\right)^2},
\end{equation}
with the arbitrary real constants $\epsilon$ and $\eta$. Obviously, the fractional rational solution \eqref{4-q1-rational} is a linear soliton with the center along the line $x{=}2(2\eta^2)^{1+\epsilon}t$. The amplitude of $|q^{[1]}_{r}(x,t)|$ is $4\eta$ as well as $|q^{[1]}(x,t)|$. Here we choose the same parameter $\eta{=}0.2$ as in Fig.\ref{fig:direction}. In Fig.\ref{fig:direction_rational}, we can observe that unlike the fractional one-soliton solution $q^{[1]}(x,t)$, this linear soliton is a right-going traveling-wave soliton. At the same time, the velocity of the traveling wave will also be faster as the increase of $\epsilon$. Note that when $|x|{\to}\infty$, $|q^{[1]}_{r}(x,t)|$ will tend to zero for arbitrary fixed $t$.
\begin{figure}[H]
	\centering
	\includegraphics[width=0.6\linewidth]{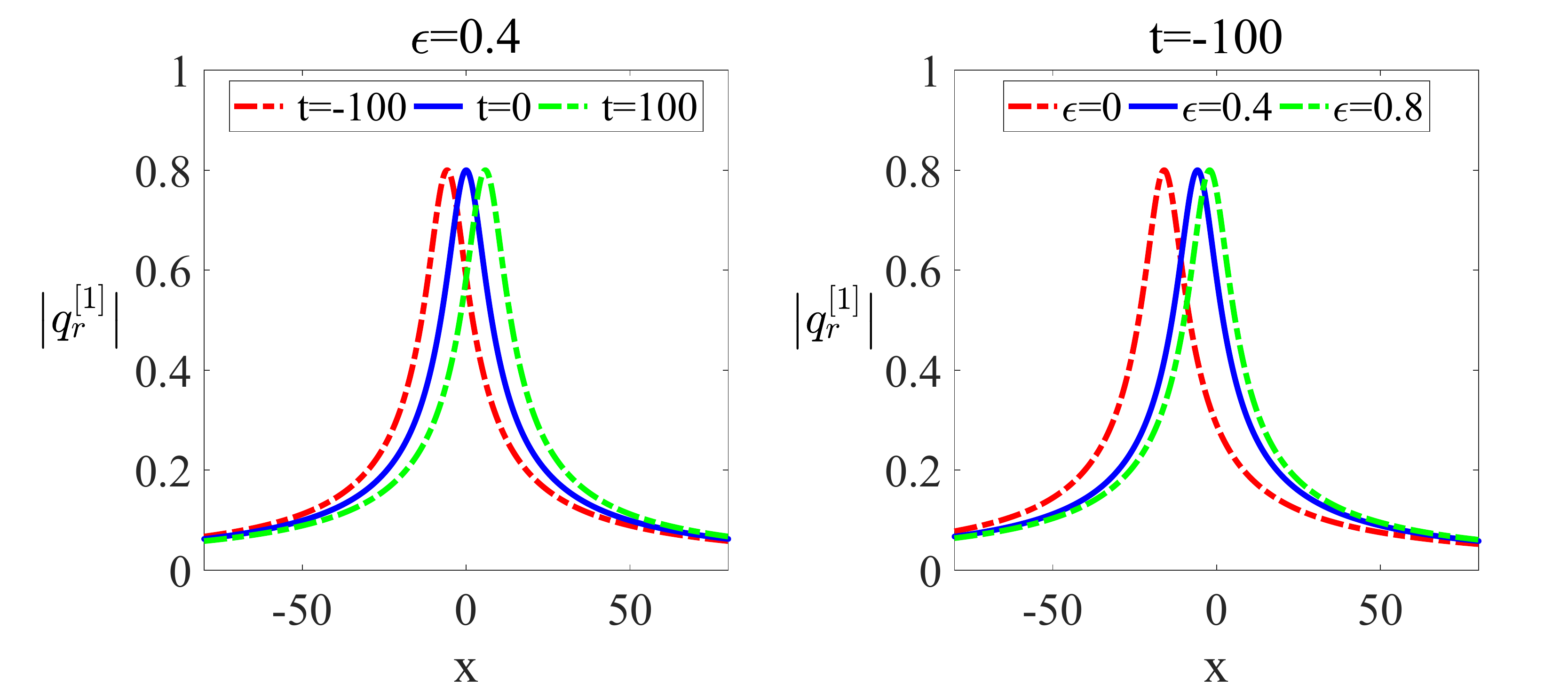}
	\caption[wave direction]{The direction of wave propagation. Choosing the parameter: $\eta=0.2.$}
	\label{fig:direction_rational}
\end{figure}

\section{Conclusion}
In conclusion, based on the fractional integrable equation of the AKNS system proposed by Ablowitz, Been, and Carr,
we extend it to the KN system in the sense of the Riesz fractional derivative. We study the fractional DNLS equation in detail and obtain the fractional $N$-soliton solution according to the IST method. In addition, we give the explicit form of the fractional one-soliton solution and provide rigorous proof by combined with the Darboux transformation in \cite{guo2013high}. And from the right panel in Fig.\ref{fig:direction}, we find that the fractional solitons propagate without dissipating or spreading out. Moreover, we discuss the limitations of the fractional one-soliton solution and get the fractional rational solution of the fDNLS equation. These phenomena will significantly enrich the dynamic properties of integrable systems and help predict the superdispersive transport of nonlinear waves in fractional nonlinear media.

Note that we only proved the fractional one-soliton solution. The fractional two-soliton solution and even the fractional $N$-soliton solution can also be verified via a similar method, but the computations will become more complicated. Therefore, finding a more convenient method to prove the solution is necessary. In addition, the nonlocal equation is also a hot research topic in the integrable system \cite{ablowitz-2017,an-2021,ling-2021}. It is worth considering whether the nonlocal and fractional integrable equations can be studied together.

\bibliographystyle{siam}
\bibliography{Ref-fDNLS}

\end{document}